%% file: main.tex
\documentclass[runningheads,a4paper,12pt]{llncs}
\usepackage{multirow}
\usepackage[T1]{fontenc}
\usepackage{graphicx}
\usepackage{float}
\usepackage{tikz}
\usepackage{fullpage}
\usepackage[symbol]{footmisc}
\usetikzlibrary{graphs}
\usetikzlibrary{positioning}
\usepackage[linesnumbered,noend]{algorithm2e}
\usepackage{amsmath}
\usepackage{amsfonts}
\usepackage{paralist}
\newcommand{\distros}{\mathcal{D}}
\newcommand{\prob}{\textsl{Pr}}
\renewcommand{\paragraph}[1]{\smallskip\noindent\emph{#1}}
\newcommand{\target}{\mathfrak{t}}
\newcommand{\targets}{\mathfrak{T}}
\newcommand{\hit}{\textsl{Hit}}
\newcommand{\hitprob}{\textsl{HitPr}}
\newcommand{\dissum}{\textsl{ExpDisSum}}

\newcommand{\expmeanpayoff}{\textsl{ExpMP}}

\newcommand{\ld}{\delta_{\lim}}
\newcommand{\expv}{\mathbb{E}}
\newcommand{\val}{\textsl{val}}
\newcommand{\w}{\textsl{w}}
\newcommand{\tw}{\textsl{tw}}
\usepackage{amssymb}
\definecolor{cadmiumgreen}{rgb}{0.0, 0.42, 0.24}
\newcommand{\one}{\hat{\textbf{1}}}
\newcommand{\equ}{\mathfrak{e}}
\renewcommand{\paragraph}[1]{\smallskip\noindent\emph{#1}}
\SetAlFnt{\normalsize}
\usepackage{hyperref}
\usepackage{adjustbox}
\usepackage{rotating}
\usepackage{array}
\newcolumntype{H}{>{\setbox0=\hbox\bgroup}c<{\egroup}@{}}

\begin{document}
\title{Faster Algorithms for Quantitative Analysis\\of Markov Chains and Markov Decision Processes\\with Small Treewidth}
\titlerunning{Faster Algorithms for MCs and MDPs with Small Treewidth}

\author{Ali Asadi\inst{1} \and Krishnendu Chatterjee\inst{2} \and Amir Kafshdar Goharshady\inst{2} \and\\ Kiarash Mohammadi\inst{3} \and Andreas Pavlogiannis\inst{4}}
\authorrunning{A. Asadi, K. Chatterjee, A.K. Goharshady, K. Mohammadi, A. Pavlogiannis}
\institute{Sharif University of Technology, Tehran, Iran, \email{asadia1376@gmail.com} \and
IST Austria, Klosterneuburg, Austria, \email{\{kchatterjee, goharshady\}@ist.ac.at}
\and Ferdowsi University of Mashhad, Mashhad, Iran, \email{kiarash.km@gmail.com}
\and Aarhus University, Aarhus, Denmark, \email{pavlogiannis@cs.au.dk}}

\maketitle             

\renewcommand{\thefootnote}{\fnsymbol{footnote}}
\setcounter{footnote}{0}

\input{abstract}

\input{introduction}

\input{preliminaries}
\input{mc}

\input{experiments}

\newpage
\bibliographystyle{splncs04}
\bibliography{refs}

\input{appendix}

\end{document}

%% file: abstract.tex
\begin{abstract} 
Discrete-time Markov Chains (MCs) and Markov Decision Processes (MDPs) are two standard formalisms in system analysis.
Their main associated \emph{quantitative} objectives are hitting probabilities, discounted sum, and mean payoff.
Although there are many techniques for computing these objectives in general MCs/MDPs, they have not been thoroughly studied in terms of parameterized algorithms, particularly when treewidth is used as the parameter. 
This is in sharp contrast to \emph{qualitative} objectives for MCs, MDPs and graph games, for which treewidth-based algorithms yield significant complexity improvements.

In this work, we show that treewidth can also be used to obtain faster algorithms for the quantitative problems.
For an MC with $n$ states and $m$ transitions, we show that each of the classical quantitative objectives can be computed in $O((n+m)\cdot t^2)$ time, given a tree decomposition of the MC that has width $t$.
Our results also imply a bound of $O(\kappa\cdot (n+m)\cdot t^2)$ for each objective on MDPs, where $\kappa$ is the number of strategy-iteration refinements required for the given input and objective.
Finally, we make an experimental evaluation of our new algorithms on low-treewidth MCs and MDPs obtained from the DaCapo benchmark suite.
Our experimental results show that on MCs and MDPs with small treewidth, our algorithms outperform existing well-established methods by one or more orders of magnitude.

\end{abstract}

%% file: introduction.tex
\section{Introduction} \label{sec:intro}

\paragraph{Markov Chains.}
Perhaps the most standard formalism for modeling randomness in discrete-time systems is that of discrete-time Markov Chains (MCs)~\cite{norris1998markov,mcbook}.
MCs have immense applications in verification, and are used to express randomness both in the system itself~\cite{Courcoubetis95,Rutten04} and in the environment that the system interacts with~\cite{Chatterjee10}.
The modeling power of MCs has also led to various extensions, such as parametric~\cite{Daws05,Lanotte07,Hahn09}, interval~\cite{Larsen91,Koushik06,Delahaye11,Benedikt13} and
augmented interval~\cite{Chonev19} MCs.
Besides the theoretical appeal, the analysis of MCs is also a core component in several model checkers~\cite{storm,Hermanns00,Hahn10,Kwiatkowska11}.

\paragraph{Markov Decision Processes.}
When the system exhibits both stochastic and non-deterministic behavior, the standard model of MCs is
lifted to Markov Decision Processes (MDPs)~\cite{bellman1957markovian,feinberg2012handbook}.
For example, MDPs are used to model stochastic controllers, where non-determinism models the freedom of the controller and randomness models the behavior of the system~\cite{Filar96}.
MDPs are also a topic of active study in verification~\cite{Bianco95,Brazdil14,Tappler19,Chatterjee18b,Hahn17}.

\paragraph{Quantitative Analysis.}
Three of the most standard analysis objectives for MCs are the following:
(a)~The \emph{hitting probabilities} objective takes as input a set of target vertices $\targets$ of the MC, 
and asks to compute for each vertex $u$, the probability that a random walk starting from $u$ will eventually hit $\targets$.
The \emph{discounted sum} objective takes as input a discount factor $\lambda\in(0,1)$ and a  reward function $R$ that assigns a reward to each edge of the MC.
The task is to compute for each vertex $u$ the expected reward value of a random walk starting from $u$,
where the value of the walk is the sum of the rewards along its edges, discounted by the factor $\lambda$ at each step.
Finally, the \emph{mean payoff} objective is similar to the discounted sum objective,
except that the value of a walk is the long-run average of the rewards along its edges.
In MDPs, the corresponding analysis questions ask for a strategy that maximizes the respective quantity.

\paragraph{Analysis Algorithms.}
Given the importance of quantitative objectives for MCs and MDPs, there have been various techniques for solving them efficiently.
For MCs, the hitting probabilities and discounted sum objectives reduce to solving a system of linear equations~\cite{mcbook,norris1998markov,bellman1957markovian,Kretinsky17}.
For MDPs, all three objectives reduce to solving a linear program~\cite{norris1998markov,Kretinsky17,papadimitriou1987complexity}.
Besides the LP formulation, two popular approaches for solving quantitative objectives on MDPs are
value iteration~\cite{bellman1957markovian,Bellman10} and strategy iteration~\cite{howard1960dynamic,Littman95,Mansour99,Abate16,puterman2014markov}.
Value iteration is the most commonly used method in verification and operates by computing optimal policies for successive finite horizons. 
However, this process leads only to approximations of the optimal values, and for some objectives no stopping criterion for the optimal strategy is known~\cite{Ashok17}. In cases where such criteria are known (e.g.~\cite{baier2017ensuring,haddad2018interval,quatmann2018sound}), the number of iterations necessary before the numbers can be rounded to provide an optimal solution can be extremely high~\cite{chatterjee2008valueiter}. Nevertheless, value iteration has proved to be very successful in practice and is included in many probabilistic model checkers, such as~\cite{Kwiatkowska11,storm}.
On the other hand, strategy iteration lies on the observation that given a fixed strategy, the MDP reduces to an MC, and hence one can compute the value of each vertex using existing techniques on MCs.
Then, the strategy can be refined to a new strategy that improves the value of each vertex.
The running time of strategy iteration can be written as $O(\kappa \cdot f)$, where $\kappa$ is the number of strategy refinements and $f$ is the time for evaluating the strategy.
As we saw above, $f$ is bounded by the time required to solve a linear system (instead of a linear program).
In addition, $\kappa$ is bounded by the number of possible strategies and thus finite, 
and although $\kappa$ can be exponentially large~\cite{Fearnley10,Hollanders12}, it behaves as a small constant in practice, which makes strategy iteration work well in practice~\cite{puterman2014markov,Kretinsky17}.
Hence, both for MCs and for MDPs using strategy iteration, the performance of the algorithm largely depends on the speed of solving the respective linear system~\cite{Kretinsky17}.

\paragraph{Treewidth.}
A very well-studied notion in graph theory is the concept of {\em treewidth} 
of a graph, which is a measure of how similar a graph is to a tree~\cite{robertson1984graph}. For example, a connected graph has treewidth~1 precisely if it is a tree. 
On one hand the treewidth property provides a mathematically elegant way 
to study graphs, and on the other hand there are many classes of graphs which 
arise in practice and have constant treewidth. 
A prime example is that Control Flow Graphs (CFGs) of \texttt{goto}-free 
programs in many classical programming languages have constant 
treewidth~\cite{thorup1998all}.
The low treewidth of flow graphs has also been confirmed experimentally 
for programs written in Java~\cite{gustedt2002treewidth}, C~\cite{Krause19}, Ada~\cite{Burgstaller04} and Solidity~\cite{chatterjee2019treewidth}.
Treewidth has important algorithmic implications, as many graph problems that are hard to solve in general admit efficient solutions on graphs of low treewidth~\cite{cygan2015parameterized}.
In program analysis, this property has been exploited to develop improvements for
register allocation~\cite{thorup1998all,Bodlaende98},
on-demand algebraic-path analysis~\cite{CIPG15},
on-demand data-flow analysis of concurrent programs~\cite{toplas}
and data-dependence analysis~\cite{Chatterjee18c}.
Treewidth has also been studied in the context of parameterized algorithms for model checking~\cite{Obsrzalek03,Ferrara05}.

\paragraph{Our Contributions.}
The contributions of this work are as follows:

\begin{compactenum}
\item{\em Theoretical Contributions.}
Our main theoretical result is a linear-time algorithm for solving arbitrary systems of linear equations whose primal graph has low treewidth.
Given a linear system $S$ of $m$ equations over $n$ unknowns, and a tree decomposition of the primal graph of $S$ that has width $t$,
our algorithm solves $S$ in time $O((n+m)\cdot t^2)$.
Given an MC $M$ of treewidth $t$ and a corresponding tree decomposition,
our algorithm directly implies similar running times for the hitting probabilities and discounted sum objectives for $M$.
In addition, we develop an algorithm that solves the mean-payoff objective for $M$ in time $O((n+m)\cdot t^2)$.
Our results on MCs also imply upper-bounds for the running time of strategy iteration on low-treewidth MDPs.
Given an MDP $P$ with treewidth $t$ and a quantitative objective,
our results imply that $P$ can be solved in time $O(\kappa \cdot (n+m)\cdot t^2)$, where $\kappa$ is the number of iterations until strategy iteration stabilizes for the respective input and objective.

\item{\em Practical Contributions.}
We develop two practical algorithms for solving the hitting probabilities and discounted sum objectives on low-treewidth MCs.
Although these algorithms have the same worst-case complexity of $O((n+m)\cdot t^2)$ as our general solution, they avoid its most practically time-consuming step, i.e.~applying the Gram-Schmidt process, and replace it with simple changes to the MC.  We report on an implementation of these algorithms and their performance in solving MCs and MDPs with low treewidth.
We perform an extensive comparison of our implementation and previous methods as follows:

\begin{compactenum}
\item \emph{Comparison with classical approaches}: We compare our algorithms for MCs against a heavily-optimized Gaussian elimination. In case of MDPs, we additionally compare with classical value-iteration and strategy-iteration methods.
\item \emph{Comparison with out-of-the-box tools}: We compare our implementation with standard industrial optimizers and probabilistic model checkers, including Matlab~\cite{MATLAB10}, lpsolve~\cite{lpsolve}, Gurobi~\cite{gurobi}, PRISM~\cite{Kwiatkowska11} and Storm~\cite{storm}.
\end{compactenum}

Our results show a consistent advantage of our new algorithms over all baseline methods, when the input models have small treewidth.
Our algorithms outperform both the existing classical approaches for solving MCs/MDPs, and the highly-refined standard solvers.
\end{compactenum}

\smallskip\noindent{\em Closest Related Works.}
To our knowledge, the existing works closest to this paper are~\cite{Chatterjee13,fomin2018fully}.
The work of~\cite{Chatterjee13} (CAV 2013) considers the maximal end-component decomposition and the almost-sure reachability set computation in low-treewidth MDPs. Note that these are both \emph{qualitative} objectives, and thus very different from the \emph{quantitative} objectives we consider here, which cannot be solved by~\cite{Chatterjee13}. Specifically, the main problem solved by~\cite{Chatterjee13} is almost-sure reachability, i.e.~reachability with probability 1, which is a very special qualitative case of computing hitting probabilities. 
The work of~\cite{fomin2018fully} develops an algorithm for solving linear systems of low treewidth.
Considering the computational complexity when applied to MCs/MDPs of treewidth $t$, the algorithms we develop in this work are a factor $t$ faster compared to~\cite{fomin2018fully}. 
On the practical side, the algorithms in~\cite{fomin2018fully} have more complicated intermediate steps, which we expect will lead to huge constant factors in the runtime of their implementations. This being said, it is highly nontrivial to provide a practically efficient implementation of~\cite{fomin2018fully} and we are not aware of any implementation for it.

%% file: preliminaries.tex
\section{Preliminaries} \label{sec:prelim}

\subsection{Markov Chains and Markov Decision Processes} 
\label{sec:mdp}

\paragraph{Discrete Probability Distributions.} Given a finite set $X$, a probability distribution over $X$ is a function $d: X \rightarrow [0, 1]$ such that $\sum_{x \in X} d(x) = 1.$ We denote the set of all probability distributions over $X$ by $\distros(X)$.

\paragraph{Markov Chains (MCs)~\textnormal{\cite{kemeny2012denumerable}}.} A \emph{Markov chain} $C = (V, E, \delta)$ consists of a finite directed graph $(V, E)$ and a probabilistic transition function $\delta: V \rightarrow \distros(V)$, such that for any pair $u, v$ of vertices, we have $\delta(u)(v)>0$ only if $(u, v) \in E.$ 

In an MC $C$, we start a random walk from a vertex $v_0 \in V$ and at each step, being in a vertex $v$, we probabilistically choose one of the successors of $v$ and go there. The probability with which a successor $w$ is chosen is given by $\delta(v)(w).$ Let $O$ be a measurable set of infinite paths on $V$ (or more generally let~$O \subseteq V^\omega$), we use the notation $\prob_{v_0}(O)$ to denote the probability that our infinite random walk starting from $v_0$ is a member of $O$. See~\cite{mcbook,kemeny2012denumerable} for more detailed treatment.

\paragraph{Markov Decision Processes (MDPs)~\textnormal{\cite{howard1960dynamic,Filar96}}.} A \emph{Markov decision process} $P = (V, E, V_1, V_P, \delta)$ consists of a finite directed graph $(V, E)$, a partitioning of $V$ into two sets $V_1$ and $V_P$, and a probabilistic transition function $\delta: V_P \rightarrow \distros(V),$ such that for any $(u, v) \in V_P \times V,$ we have $\delta(u)(v)>0$ only if $(u, v) \in E.$  In this work, we assume that all vertices of an MDP have at least one outgoing edge.
Intuitively, an MDP is a one-player game in which we have two types of vertices: those controlled by Player 1, i.e.~$V_1$, and those that behave probabilistically, i.e.~$V_P$.

\paragraph{Strategies.} In an MDP $P$, a \emph{strategy}  is a function $\sigma: V_1 \rightarrow V$, such that for every $v \in V_1$ we have $(v, \sigma(v)) \in E.$\footnote{We are only considering pure memoryless strategies because they are sufficient for our use-cases, i.e.~there always exists an optimal strategy that is both pure and memoryless~\cite{Filar96,Kretinsky17}.}

Informally, a strategy is a recipe for Player 1 that tells her which successor to choose based on the current state (vertex). Given an MDP $P$ with a strategy $\sigma$, we start a random walk from a vertex $v_0 \in V$ and at each step, being in a vertex $v$, choose the successor as follows: (i) if $v \in V_1$, then we go to $\sigma(v)$, and (ii) if $v \in V_P$ we act as in the case of MCs, i.e. we go to each successor $w$ with probability $\delta(v)(w).$  As before, given a measurable set $O \subseteq V^\omega$ of infinite paths on $V$, we define $\prob^\sigma_{v_0}(O)$ as the probability that our infinite random walk becomes a member of $O$.  Note that an MDP with a fixed strategy $\sigma$ is basically an MC, in which for every $v \in V_1$ we have $\delta(v)(\sigma(v)) = 1$. See~\cite{Filar96,howard1960dynamic} for more details.

\paragraph{Hitting Probabilities~\textnormal{\cite{grinstead2012introduction,norris1998markov,krak2019hitting}}.} Let $C = (V, E, \delta)$ be an MC and $\targets \subseteq V$ a designated set of \emph{target} vertices. We define $\hit(\targets) \subseteq V^\omega$ as the set of all infinite sequences of vertices that intersect $\targets$. The \emph{Hitting probability} $\hitprob(u, \targets)$ is defined as $\prob_u(\hit(\targets)).$ In other words,  $\hitprob(u, \targets)$ is the probability of eventually reaching $\targets$, assuming that we start our random walk at $u$. In case of MDPs, we assume that the player aims to maximize the hitting probability by choosing the best possible strategy. Therefore, in an MDP $P = (V, E, V_1, V_P, \delta),$ we define $\hitprob(u, \targets)$ as $\max_{\sigma} \prob^\sigma_u(\hit(\targets)).$

\paragraph{Discounted Sums of Rewards~\textnormal{\cite{puterman2014markov}}.} Let $C = (V, E, \delta)$ be an MC and $R : E \rightarrow \mathbb{R}$ a \emph{reward function} that assigns a real value to each edge. Also, let $\lambda \in (0, 1)$ be a \emph{discount factor}. Given an infinite path $\pi = v_0, v_1, \ldots$ over $(V, E)$, we define the total reward  $R(\pi)$ of $\pi$ as $\sum_{i=0}^\infty \lambda^i \cdot R(v_i, v_{i+1}) = R(v_0, v_1) + \lambda \cdot R(v_1, v_2) + \lambda^2 \cdot R(v_2, v_3) + \ldots$. Let $u \in V$ be a vertex, we define $\dissum(u)$ as the expected value of the reward of our random walk if we begin it at $u$, i.e.~$\dissum(u) := \expv_u[R(\pi)]$. As in the previous case, when considering MDPs, we assume that the player aims to maximize the discounted sum, hence given an MDP $P = (V, E, V_1, V_P, \delta)$, a reward function $R$ and a discount factor $\lambda$, we define $\dissum(u) := \max_{\sigma} \expv^\sigma_u[R(\pi)].$

\paragraph{Mean Payoff~\textnormal{\cite{puterman2014markov,Kretinsky17}}.} Let $C$ be an MC and $R$ a reward function as above. Given an infinite path $\pi=v_0, v_1, \ldots$ over $C$, we define the $n$-step average reward of $\pi$ as $R(\pi[0..n]) := \frac{1}{n} \sum_{i=1}^n R(v_{i-1}, v_i).$ Given a start vertex $u \in V,$ the expected \emph{long-time average} or \emph{mean payoff} value from $u$ is defined as $\expmeanpayoff(u) := \lim_{n \rightarrow \infty}\expv_u [R(\pi[0..n])].$ In other words, $\expmeanpayoff(u)$ captures the expected reward per step in a random walk starting at $u$.  As in previous cases, in an MDP $P$, we define $\expmeanpayoff(u) := \max_\sigma \lim_{n \rightarrow \infty}\expv^\sigma_u [R(\pi[0..n])].$ The limits in the former definitions are guaranteed to exist~\cite{puterman2014markov,Kretinsky17}.

\paragraph{Problems.} We consider the following classical problems for both MCs and MDPs:
\begin{compactitem}
	\item \emph{Computing Hitting Probabilities:}
		 Given an MC/MDP and a target set $\targets$ compute $\hitprob(u, \targets)$ for every vertex $u.$
		 
 	\item \emph{Computing Expected Discounted Sums:}
 	 Given an MC/MDP, a reward function $R$ and a discount factor $\lambda \in (0, 1),$ compute $\dissum(u)$ for every vertex $u$.
 	 
 	 \item \emph{Computing Mean Payoffs:} Given an MC/MDP and a reward function $R$, compute $\expmeanpayoff(u)$ for every vertex $u$.
 		
\end{compactitem}

\paragraph{Solving MCs~\textnormal{\cite{grinstead2012introduction,norris1998markov}}.} A classical approach to the above problems for MCs is to reduce them to solving systems of linear equations.  In case of hitting probabilities, we define one variable $x_u$ for each vertex $u$, whose value in the solution to the system would be equal to $\hitprob(u, \targets).$ The system is constructed as follows:
\begin{compactitem}
	\item We add the equation $x_\target = 1$ for every $\target \in \targets$, and
	\item For every vertex $u \not\in \targets$ with successors $u_1, \ldots, u_k$, we add the equation $x_u = \sum_{i=1}^k \delta(u)(u_i) \cdot x_{u_i}.$
\end{compactitem} 
If every vertex can reach a target, then it is well-known that the resulting system has a unique solution in which the value assigned to each $x_u$ is equal to $\hitprob(u, \targets)$\footnote{Otherwise, we can first remove the vertices that cannot reach a target by a simple DFS and then apply the algorithm to the rest of the MC.}. A similar approach can be used in the case of discounted sums. We define one variable $y_u$ per vertex $u$ and if the successors of $u$ are $u_1, \ldots, u_k$, then we add the equation $y_u = \sum_{i=1}^k \delta(u)(u_i) \cdot \left( R(u, u_i) + \lambda \cdot y_{u_i} \right)$. The approach for mean payoff objectives is more subtle and described in Section~\ref{sec:mp}.

\paragraph{Primal Graphs~\textnormal{~\cite{rossi2006handbook}.}} Let $S$ be a system of linear equations with $m$ equations and $n$ unknowns (variables). The primal graph $G(S)$ of $S$ is an undirected graph with $n$ vertices, each corresponding to one unknown in $S$, in which there is an edge between two unknowns $x$ and $y$ iff there exists an equation in $S$ that contains both $x$ and $y$ with non-zero coefficients. 

\paragraph{Solving MDPs.}  There are two classical approaches to solving the above problems for MDPs. One is to reduce the problem to Linear Programming (LP) in a manner similar to the reduction from MC to linear systems~\cite{feinberg2012handbook}. The other approach is to use dynamic programming~\cite{bellman1957markovian,shapley1953stochastic}. We consider a widely-used variety of dynamic programming, called \emph{strategy iteration} or \emph{policy iteration}~\cite{howard1960dynamic,bertsekas1995dynamic}.

\paragraph{Strategy Iteration (SI)~\textnormal{\cite{bellman1957markovian,shapley1953stochastic}}.} In SI we start with an arbitrary initial strategy $\sigma_0$ and attempt to find a better strategy in each step. Formally, assume that our strategy after $i$ iterations is $\sigma_i$. Then, we compute $\val_i(u) = \hitprob^{\sigma_i}(u, \targets)$ for every vertex $u$. This is equivalent to computing hitting probabilities in the MC that is obtained by considering our MDP together with the strategy $\sigma_i.$ We use the values $\val_i(u)$ to obtain a better strategy $\sigma_{i+1}$ as follows: for every vertex $v \in V_1$ with successors $v_1, v_2, \ldots, v_k$, we set $\sigma_{i+1}(v) = \arg\max_{v_j} \val_i(v_j)$. (In case of discounted sum, we let $\val_i(u) = \dissum^{\sigma_i}(u)$ and $\sigma_{i+1}(v) = \arg\max_{v_j} R(v, v_j) + \lambda \cdot \val_i(v_j)$.) We repeat these steps until we reach a point where our strategy converges, i.e. it does not change anymore. It is well-known that strategy iteration always converges to the optimal strategy, and at that point the values $\val_i$ will be the desired hitting probabilities/discounted sums \cite{howard1960dynamic,bertsekas1995dynamic,feinberg2012handbook}. Moreover, while it might take exponentially many steps in theory~\cite{fearnley2010strategy,hansen2012worst}, SI is one of the most practical algorithms for solving MDPs and almost always terminates within a few iterations in real-world scenarios~\cite{hansen2012worst,Kretinsky17}. Hence, a major challenge is to optimize the runtime of each iteration~\cite{Kretinsky17}. SI can also be applied to mean payoff objectives. However, it requires the computation of additional values, called \emph{potentials} or \emph{biases}. See~\cite{Kretinsky17,puterman2014markov} for more details.

Given that SI solves the classic problems above on MDPs by several calls to a procedure for solving the same problems on MCs, our runtime improvements for MCs are naturally extended to MDPs. So, in the sequel we turn our focus to MCs. 

\subsection{Parameterized Algorithms, Tree Decompositions and Treewidth}

\paragraph{Parameterized Complexity~\textnormal{\cite{downey2012parameterized}}.} In parameterized complexity, the runtime of an algorithm is analyzed not only based on the size of its input, but also based on an aspect of the input, called a ``parameter''. Hence, parameterized complexity provides a finer-grained understanding than traditional complexity theory. For example, given a graph $G$ and an integer $k$ as input, it is NP-hard to decide whether $G$ has a vertex cover\footnote{A vertex cover is a set $C$ of vertices such that each edge has at least one of its endpoints in $C$.} of size $k$. However, there is an algorithm with runtime  $O(n \cdot k \cdot 2^k)$ for this problem~\cite{cygan2015parameterized}. Hence, if $k$ is a small constant, then the problem is solvable in linear time. The parameter does not necessarily need to be an explicit part of the input. It can also be a structural property of the input instance, e.g.~many hard graph problems are efficiently solvable over graphs whose maximum degree is small~\cite{cygan2015parameterized}.

\paragraph{Fixed-Parameter Tractability (FPT)~\textnormal{\cite{downey2012parameterized,cygan2015parameterized}}.} A parameterized problem is called \emph{fixed parameter tractable} if it can be solved in time $O(n^c \cdot f(k)),$ where $n$ is the input size, $k$ is the parameter, $f$ is an arbitrary computable function and $c$ is a constant that is not dependent on either $n$ or $k$.  This definition captures the intuition that while the problem might be hard in general, those instances of the problem where the parameter is small are easy to solve, i.e.~they are solvable in polynomial time wrt the size of input\footnote{Note that the polynomial degree $c$ is not dependent on the parameter $k$.}. In this work, we provide linear parameterized algorithms for MCs. In other words, in all of our algorithms we have $c = 1.$

Treewidth~\cite{robertson1984graph} is a widely-used parameter for graph problems. Intuitively, the treewidth of a graph is a measure of its tree-likeness, e.g.~only trees and forests have a treewidth of~$1$. We now provide a formal definition for treewidth based on tree decompositions.

\begin{figure}[H]
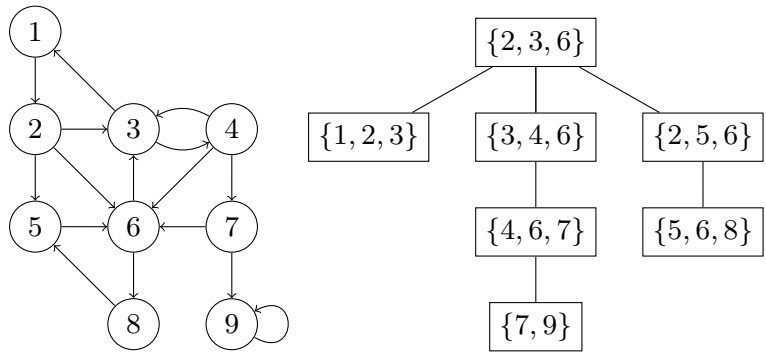

	\begin{minipage}{0.48\linewidth}
		\begin{flushright}
		\include{fig-ex1}
		\end{flushright}
	\end{minipage}
\begin{minipage}{0.48\linewidth}
	\begin{flushleft}
	\include{fig-ex1-td}
	\end{flushleft}
\end{minipage}
\caption{A graph $G$ (left) and a tree decomposition of $G$ with width $2$ (right).}
\label{fig:td-example}
\end{figure}

\paragraph{Tree Decompositions~\textnormal{\cite{robertson1984graph,bodlaender1994tourist}}.} Given a directed or undirected graph $G = (V, E)$, a \emph{tree decomposition} of $G$ is a tree $(T, E_T)$ such that:
\begin{compactitem}
	\item Each vertex $b \in T$ of the tree is associated with a subset $V_b \subseteq V$ of vertices of the graph. For clarity, we reserve the word ``vertex'' for vertices of $G$ and use the word ``bag'' to refer to vertices of $T$. Also, we define $E_b := \{ (u, v) \in E ~\vert~ u, v \in V_b \}.$
	
	\item Each vertex appears in at least one bag, i.e.~$\bigcup_{b \in T} V_b = V.$
	
	\item Each edge appears in at least one bag, i.e.~$\bigcup_{b \in T} E_b = E.$
	
	\item Each vertex appears in a connected subtree of $T$. In other words, for all $b, b', b'' \in T$, if $b''$ is in the unique path between $b$ and $b'$, then $V_{b} \cap V_{b'} \subseteq V_{b''}.$
\end{compactitem}

\paragraph{Treewidth~\textnormal{\cite{robertson1984graph,cygan2015parameterized}}.} The \emph{width} of a tree decomposition is the size of its largest bag minus one, i.e.~$\w(T) = \max_{b \in T} \vert V_b \vert - 1.$ A tree decomposition of $G$ is called \emph{optimal} if its width is less than or equal to the width of any other tree decomposition. The \emph{treewidth} $\tw(G)$ of $G$ is defined as the width of its optimal tree decomposition(s).

\paragraph{Computing Treewidth and Tree Decompositions.} Computing treewidth is an NP-complete problem~\cite{arnborg1987complexity}. However, it is solvable in linear-time FPT wrt the treewidth itself, i.e.~if we know that $\tw(G)$ is bounded by a constant, then the problem is solvable in linear time~\cite{bodlaender1996linear}. In this case, the algorithm in~\cite{bodlaender1996linear} also finds an optimal tree decomposition in linear time\footnote{Specifically, note that~\cite{thorup1998all} proves that control-flow graphs of structured programs in C and Pascal have a treewidth of at most $6$.}.
In the sequel, we focus on linear-time algorithms for MCs and MDPs parameterized by their treewidth. As is standard for treewidth-based approaches, we assume that an optimal tree decomposition is given as part of the input. 
This assumption does not affect the complexity of our approach, as we can use~\cite{bodlaender1996linear,thorup1998all} or tools such as~\cite{jtdec} to obtain the tree decomposition in linear time.

%% file: fig-ex1.tex
\begin{tikzpicture}[scale=1.2, transform shape, node distance = {5mm and 5mm}, v/.style = {draw, circle, minimum size=5mm}, b/.style = {draw, rectangle}]
	
	\node (1) [v] {1};
	\node (2) [v, below = of 1] {2};
	\node (3) [v, right = of 2] {3};
	\node (4)  [v, right = of 3] {4};
	\node (5)  [v, below = of 2] {5};
	\node (6)  [v, right = of 5] {6};
	\node (7)  [v, right = of 6] {7};
	\node (8)  [v, below = of 6] {8};
	\node (9)  [v, right = of 8] {9};
	
	\draw (1) [->] to (2);
	\draw (1) [<-] to (3);
	\draw (2) [->] to (3);
	\draw (2) [->] to (5);
	\draw (2) [->] to (6);
	\draw (3) [->, bend right] to (4);
	\draw (4) [->, bend right] to (3);
	\draw (3) [<-] to (6);
	\draw (4) [->] to (6);
	\draw (4) [->] to (7);
	\draw (5) [->] to (6);
	\draw (5) [<-] to (8);
	\draw (6) [<-] to (7);
	\draw (6) [->] to (8);
	\draw (7) [->] to (9);
	\path (9) edge [->, out=330,in=30,looseness=5] (9);
	
\end{tikzpicture}

%% file: fig-ex1-td.tex
\begin{tikzpicture}[scale=1.2, transform shape, node distance = {5mm and 5mm}, v/.style = {draw, circle}, b/.style = {draw, rectangle}]
	
	\node (B) [b] {$\{2, 3, 6\}$};
	\node (D) [b, below = of B] {$\{3, 4, 6\}$};
	\node (C) [b, right = of D] {$\{2, 5, 6\}$};
	\node (A) [b, left = of D] {$\{1, 2, 3\}$};
	\node (E) [b, below = of C] {$\{5, 6, 8\}$};
	\node (F) [b, below = of D] {$\{4, 6, 7\}$};
	\node (G) [b, below = of F] {$\{7, 9\}$};
	
	\draw (A) to (B);
	\draw (B) to (C);
	\draw (B) to (D);
	\draw (C) to (E);
	\draw (B) to (D);
	\draw (D) to (F);
	\draw (F) to (G);
	
\end{tikzpicture}

%% file: mc.tex
\section{Algorithms for MCs with Constant Treewidth} \label{sec:mc}

In this section, we consider quantitative problems on MCs. As mentioned before, our improvements carry over to MDPs using SI. We build on classical state-elimination algorithms to handle our MCs. Such methods are well-known and were previously used  in~\cite{hopcroft2001introduction,daws2004symbolic,hahn2010param,hahn2011probabilistic}, as well as many other works. The main novelty of our approach is that we use the tree decompositions to obtain a suitable \emph{order} for eliminating vertices. This specific ordering significantly reduces the runtime complexity of classical state-elimination algorithms from cubic to linear. Aside from the ordering, which is the main basis for our algorithmic improvements, the rest of this section consists mostly of well-known transformations on MCs. 
However, a new subtlety arises in our approach: while in general MCs there are several variants of rules for eliminating vertices, in small-treewidth MCs we must also make sure that the elimination step does not increase the treewidth or invalidate the underlying tree decomposition.

\smallskip

We first review state-elimination for computing hitting probabilities (Section~\ref{subsec:hp}). Then, in Section~\ref{sec:hitprobtd}, we show how to exploit the treewidth to speedup this process and obtain a linear-time algorithm. Section~\ref{subsec:discounted} provides a similar speedup for computing expected discounted sums. In Section~\ref{sec:system}, we show our most general result, i.e.~solving small-treewidth systems of linear equations in linear time. While this algorithm is more general than those of Sections~\ref{sec:hitprobtd} and~\ref{subsec:discounted}, it repeatedly applies the costly Gram-Schmidt orthogonalization process, and is hence not preferable in practice. 
Finally, Section~\ref{sec:mp} combines these ideas to compute expected mean payoffs in linear time.

\subsection{A Simple Algorithm for Computing Hitting Probabilities}\label{subsec:hp}

	We begin by looking into the problem of computing hitting probabilities for general MCs without exploiting the treewidth. First, note that, without loss of generality, we can assume that our target set contains a single vertex. Otherwise, we add a new vertex $\target$ and add edges with probability $1$ from every target vertex to $\target.$ This will keep the hitting probabilities intact.
	
	Consider our Markov chain $C = (V, E, \delta)$ and our target vertex $\target \in V.$ If there is only one vertex in the MC, i.e.~if $V = \{\target\},$ then there is not much to solve. We just return that $\hitprob(\target, \target) = 1.$ Otherwise, we take an arbitrary vertex $u \neq \target$ and try to remove it from the MC in order to obtain a smaller MC that can in turn be solved using the same method. We should do this in a manner that does not change $\hitprob(v, \target)$ for any vertex $v \neq u.$
Figure~\ref{fig:vertex-rm} shows how to remove a vertex $u$ from $C$ in order to obtain a smaller MC $\overline{C} = (V \setminus \{u\}, \overline{E}, \overline{\delta})$\footnote{In the sequel, we always use $\overline{C}$ to denote an MC that is obtained from $C$ by removing one vertex. We also apply the same rule across our notation, e.g.~$\overline{\delta}$ is the transition function after removal of the vertex.}. Basically, we remove $u$ and all of its edges, and instead add new edges from every predecessor $u'$ to every successor $u''.$ We also update the transition function $\delta$ by setting $\overline{\delta}(u')(u'') = \delta(u')(u'') + \delta(u')(u) \cdot \delta(u)(u'')$. It is easy to verify that for every $v \neq u,$ we have $\overline{\hitprob}(v, \target) = \hitprob(v, \target).$ Hence, we can compute hitting probabilities for every vertex $v \neq u$ in $\overline{C}$ instead of $C$. Finally, if $u_1, u_2, \ldots, u_k$ are the successors of $u$ in $C$, we know that $\hitprob(u, \target) = \sum_{i=1}^k \delta(u)(u_i) \cdot \hitprob(u_i, \target) = \sum_{i=1}^k \delta(u)(u_i) \cdot \overline{\hitprob}(u_i, \target).$ Hence, we can easily compute the hitting probability for $u$ using this formula. A pseudocode of this approach is available in Appendix~\ref{app:pseudo}.

A special case arises when there is a self-loop transition from $u$ to $u$. If $\delta(u)(u) = 1$, i.e.~$u$ is an absorbing trap, then we can simply remove $u$, noting that $\hitprob(u, \target) = 0$. On the other hand if $0 < \delta(u)(u) < 1,$ then we should distribute $\delta(u)(u)$ proportionately among the other successors of $u$ because staying for a finite number of steps in the same vertex $u$ does not change the hitting property of a path, and the probability of staying at $u$ forever is $0$.

\begin{figure}[H]
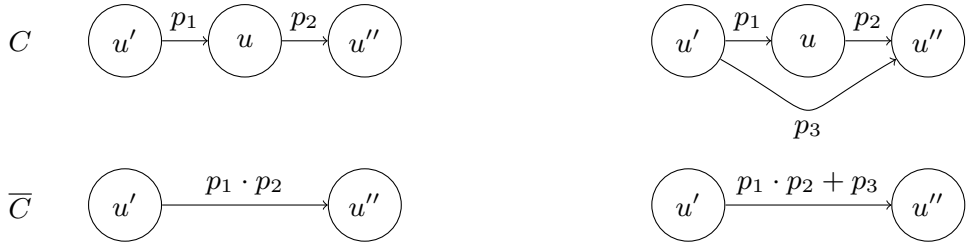

	\begin{center}
		\include{fig-vertex-rm}
	\end{center}
	\caption{Removing a vertex $u$. The vertex $u'$ is a predecessor of $u$ and $u''$ is one of its successors. The left side shows the changes when there is no edge from $u'$ to $u''$ and the right side shows the other case, where $(u', u'') \in E.$ Edge labels are $\delta$ values.}
	\label{fig:vertex-rm}
\end{figure}

Note that removing each vertex can take at most $O(n^2)$ time, given that it has $O(n)$ predecessors and successors. Using this algorithm we should remove $n-1$ vertices, leading to a total runtime of $O(n^3)$, which is worse than the reduction to system of linear equations and then applying Gaussian elimination, leading to a runtime of $O(n^\omega)$\footnote{$\omega$ is the matrix multiplcation constant, i.e.~the infimum number for which there is an algorithm that multiplies two $n \times n$ matrices in $O(n^\omega).$ We know $\omega \leq 2.373$~\cite{omega}.}. However, the runtime can be significantly improved if we could remove vertices in an order that guaranteed that every vertex has a low degree when it is being removed. One heuristic is to always remove the vertex with the smallest degree, but this does not guarantee that all removals remain cheap. %

\subsection{Computing Hitting Probabilities in Constant Treewidth} \label{sec:hitprobtd}

The main idea behind our algorithm for computing hitting probabilities in constant treewidth is very simple: we take the algorithm from the previous section and use the tree decomposition to obtain an ordering for the removal of vertices. 

Given that we can choose any bag in $T$ as the root, without loss of generality, we assume that the target vertex $\target$ is in the root bag~\footnote{If $\vert \targets \vert \geq 2$, we use the same technique as in the previous section to have only one target $\target$. To keep the tree decomposition valid, we add $\target$ to every bag.}. The following two lemmas are the bases of our approach:

\begin{lemma} \label{lemma:leaf-remove}
	Let $l \in T$ be a \emph{leaf} bag of the tree decomposition $(T, E_T)$ of our MC $C$, and let $\bar{l}$ be the parent of $l$. If $V_l \subseteq V_{\bar{l}},$ then $(T \setminus \{l\}, E_T \setminus \{(\bar{l}, l)\})$ is also a valid tree decomposition for $C$. 
\end{lemma}
\begin{proof}
	We just need to check that all the required properties of a tree decomposition hold after removal of $l$. Given that $V_l \subseteq V_{\bar{l}},$ any vertex that appears in $l$ is also in $\bar{l}$ and hence removal of $l$ does not cause any vertex to be unrepresented in the tree decomposition. The same applies to edges. Moreover, removing a \emph{leaf} bag cannot disconnect the previously-connected set of bags containing a vertex.
\end{proof}

\begin{lemma} \label{lemma:vertex-remove}
	Let $l \in T$ be a bag of the tree decomposition  $(T, E_T)$ and assume that the vertex $u \in V$ only appears in $V_l$, i.e.~it does not appear in the vertex set of any other bag. Then, the vertex $u$ has at most $\vert V_l \vert$ predecessors/successors in $C$.
\end{lemma}
\begin{proof}
	If $u'$ is a predecessor/successor of $u$, then there is an edge between them. By definition, a tree decomposition should cover every edge. Hence, there should be a bag $b$ such that $u, u' \in V_b.$ By assumption, $u$ only appears in $V_l$. Hence, every predecessor/successor $u'$ must also appear in $V_l.$
\end{proof}

The two lemmas above give us a convenient order for removing vertices. At each step, we choose an arbitrary \emph{leaf} bag $l$. If there is a vertex $u$ that only appears in $V_l$, then we remove $u$. In this case, Lemma~\ref{lemma:vertex-remove} guarantees that $u$ has $O(t)$ predecessors and successors. Otherwise,  $V_l \subseteq V_{\bar{l}}$ (recall that each vertex appears in a connected subtree) and we can remove $l$ from our tree decomposition according to Lemma~\ref{lemma:leaf-remove}. Algorithm~\ref{algo:vertex-rem-td} puts all these steps together. Note that throughout this algorithm the tree decomposition remains valid, because we are only adding edges between vertices that are already in the same leaf bag $l$. Given that we remove at most $O(n)$ bags and $n-1$ vertices and that removing each vertex takes only $O(t^2),$ the total runtime of Algorithm~\ref{algo:vertex-rem-td} is $O(n \cdot t^2)$. Hence, we have the following theorem:

\begin{theorem} \label{thm:hitting}
	Given an MC with $n$ vertices and treewidth $t$ and an optimal tree decomposition of the MC, Algorithm~\ref{algo:vertex-rem-td} computes hitting probabilities from every vertex to a designated target set in $O(n \cdot t^2).$
\end{theorem}

\begin{algorithm}[H]
	\DontPrintSemicolon
	\SetKw{True}{true}
	\SetKw{Break}{break}
	\SetKwFunction{FMain}{ComputeHitProbs}
	\SetKwBlock{Repeat}{repeat}{}
	\SetKwProg{Fn}{Function}{:}{}
	\Fn{\FMain{$C = (V, E, \delta), \target, (T, E_T)$}}{
		\eIf{$V = \{\target\}$}
		{
			$\hitprob(\target, \target) \leftarrow 1$\;
		}
		{
			
			\Repeat
			{
				Choose an arbitrary leaf bag $l \in T$\;
				$\bar{l} \leftarrow$ parent of $l$\;
				\eIf{$V_l \subseteq V_{\bar{l}}$}
				{
					$T \leftarrow T \setminus \{l\}$\;
					$E_T \leftarrow E_T \setminus \{ (\bar{l}, l) \}$
				}
				{
					Choose an arbitrary $u \in V_l \setminus V_{\bar{l}}$\;
					$V_l \leftarrow V_l \setminus \{u\}$\;
					\Break
				}
			}
			
			\eIf{$\delta(u)(u)  = 1$}{
				$\hitprob(u, \target) \leftarrow 0$\;
				\FMain($(V \setminus \{u\}, E, \delta), \target$)\;
			}
			{
				
				$f \leftarrow \frac{1}{1 - \delta(u)(u)}$\;
				$\delta(u)(u) \leftarrow 0$\;
				$E \leftarrow E \setminus \{(u, u)\}$
				
				\ForEach{$u'' \in V_l: (u, u'') \in E$}
				{
					$\delta(u)(u'') \leftarrow \delta(u)(u'') \cdot f$\;
				}
				
				\ForEach{$u' \in V_l: (u', u) \in E$}
				{
					\ForEach{$u'' \in V_l: (u, u'') \in E$}
					{
						$\delta(u')(u'') \leftarrow \delta(u')(u'') + \delta(u')(u) \cdot \delta(u, u'')$\;
						$E \leftarrow E \cup \{ (u', u'') \}$
					}
				}
				\FMain($(V \setminus \{u\}, E, \delta), \target$)\;
				$\hitprob(u, \target) \leftarrow 0$\;
				\ForEach{$u'' \in V_l : (u, u'') \in E$}
				{
					$\hitprob(u, \target) \leftarrow \hitprob(u, \target) + \delta(u, u'') \cdot \hitprob(u'', \target)$\;
				}
			}
		}
	}
	\vspace{2mm}
	\caption{Computing Hitting Probabilities using a Tree Decomposition. }
	\label{algo:vertex-rem-td}
\end{algorithm}

\begin{example}
	Consider the graph and tree decomposition in Figure~\ref{fig:td-example} with an arbitrary transition probability function $\delta$ and target vertex $\target = 6$. On this example, Algorithm~\ref{algo:vertex-rem-td} would first choose an arbitrary leaf bag, say $\{7, 9\}$ and then realize that $9$ has only appeared in this bag. Hence it removes vertex $9$ from the MC using the same procedure as in the previous section. In the next iteration, it chooses the bag $\{7\}$ and realizes that the set of vertices in this bag is a subset of vertices that appear in its parent. Hence, it removes this unnecessary bag. The algorithm continues similarly, until only the target vertex $6$ remains, at which point the problem is trivial. Figure~\ref{fig:long-example} shows all the steps of our algorithm. Note that because the width of our tree decomposition is $2$, at each step when we are removing a vertex $u$, it has at most $3$ neighbors (counting itself).
\end{example}

\begin{figure}
	\setlength{\tabcolsep}{4pt}
	\renewcommand{\arraystretch}{1.5}
	\begin{center}
	\begin{tabular}{l|l}
		\input{fig-long-example-1} & 	\input{fig-long-example-2} \\ \hline
		\input{fig-long-example-3} & 	\input{fig-long-example-4} \\ \hline
		\input{fig-long-example-5} & \input{fig-long-example-6} \\ \hline
		\input{fig-long-example-7} & \input{fig-long-example-8} \\ \hline
		\input{fig-long-example-9} &  \input{fig-long-example-10} \\ \hline
		\input{fig-long-example-11} & \input{fig-long-example-12} \\ \hline
		\input{fig-long-example-13} & \input{fig-long-example-14} \\ \hline
	\end{tabular}
	\end{center}
	
	\caption{The Steps Taken by Algorithm~\ref{algo:vertex-rem-td} on the Graph and Tree Decomposition in Figure~\ref{fig:td-example}. The target vertex $\target = 6$ is shown in green. At each step the vertex/bag that is being removed is shown in red. An active bag whose vertices, but not itself, are considered for removal is shown in blue. After removing vertex $2$, the graph has only one vertex and the base case of the algorithm is run.}
	\label{fig:long-example}
\end{figure}
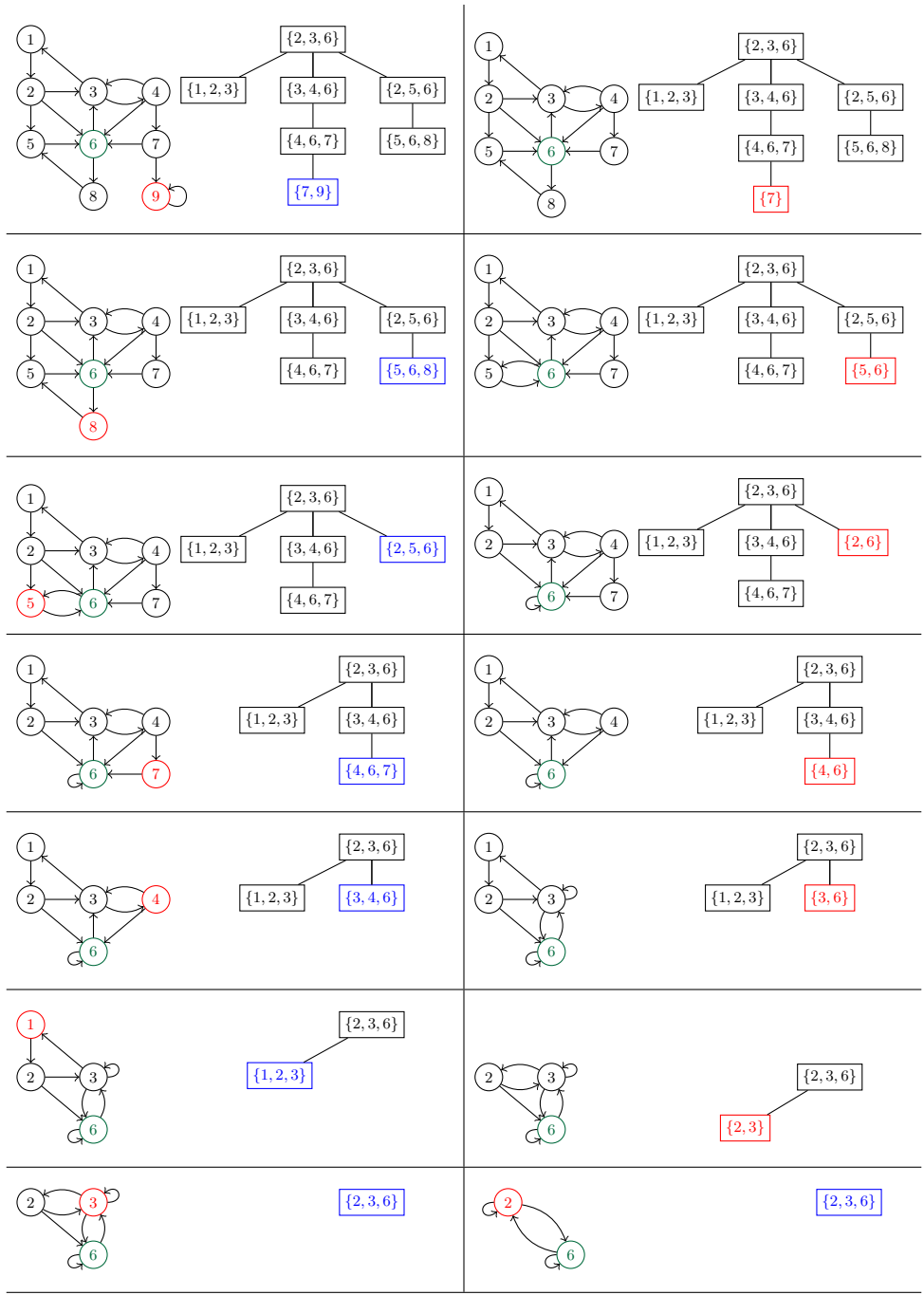

\subsection{Computing Expected Discounted Sums in Constant Treewidth}\label{subsec:discounted}

We use a similar approach for handling the discounted sum problem. The only difference is in how a vertex is removed. Given an MC $C = (V, E, \delta)$, a tree decomposition $(T, E_T)$ of $C$, a reward function $R: E \rightarrow \mathbb{R}$ and a discount factor $\lambda \in (0, 1)$, we first add a new vertex called $\one$ to the MC. The vertex $\one$ is disjoint from all other vertices and only has a single self-loop with probability $1$ and reward $1-\lambda.$ In other words, we define $\delta(\one)(\one) = 1$ and $R(\one, \one) = 1-\lambda.$
\begin{center}
\begin{tabular}{c}
	\input{fig-one}
\end{tabular}
\end{center}
We also add $\one$ to the vertex set of every bag. The reason behind this gadget is that we have $\dissum(\one) = (1-\lambda) \cdot (1 + \lambda + \lambda^2 + \ldots) = 1.$ We will use this property later.%

In our algorithm, the requirement that for all $u, v$ we should have $0 \leq \delta(u)(v) \leq 1$ is unnecessary and becomes untenable, too. Therefore, we allow $\delta(u)(v)$ to have any real value, and use the linear system interpretation of $C$ as in Section~\ref{sec:mdp}, i.e.~instead of considering $C$ as an MC, we consider it to be a representation of the linear system $S_C$ defined as follows:
\begin{compactitem}
	\item For every vertex $u \in V$, the system $S_C$ contains one unknown $y_u,$ and
	\item For every vertex $u \in V$, whose successors are $u_1, u_2, \ldots, u_k$, the system $S_C$ contains an equation $\equ_u := y_u = \sum_{i=1}^k \delta(u) (u_i) \cdot \left( R(u, u_i) + \lambda \cdot y_{u_i} \right).$
\end{compactitem}
As mentioned in Section~\ref{sec:mdp}, in the solution to $S_C$, the value assigned to the unknown $y_u$ is equal to $\dissum(u)$ in the MC $C$. However, the definition above does not depend on the fact that $C$ is an MC and can also be applied if $\delta$ has arbitrary real values.

Now suppose that we want to remove a vertex $u \neq \one$ with successors $u_1, \ldots, u_k$ from $C$. This is equivalent to removing $y_u$ from $S_C$ without changing the values of other unknowns in the solution. Given that we have $y_u = \sum_{i=1}^k \delta(u)(u_i) \cdot \left( R(u, u_i) + \lambda \cdot y_{u_i} \right),$ we can simply replace every occurrence of $y_u$ in other equations with the right-hand-side expression of this equation. If $u' \neq u$ is a predecessor of $u$, then we have $y_{u'} = A + \delta(u')(u) \cdot \left( R(u', u) + \lambda \cdot y_u \right),$ where $A$ is an expression that depends on other successors of $u'$. We can rewrite this equation as $y_{u'} = A + \delta(u')(u) \cdot R(u', u) + \sum_{i=1}^k \delta(u')(u) \cdot \delta(u)(u_i) \cdot \lambda \cdot \left( R(u, u_i) + \lambda \cdot y_{u_i} \right)$. This is equivalent to obtaining a new $\overline{C}$ from $C$ by removing the vertex $u$ and adding the following edges from every predecessor $u'$ of $u$:
\begin{compactitem}
	\item An edge $(u', \one)$, such that $R(u', \one) = 0$ and $\delta(u')(1) = \frac{1}{\lambda} \cdot (\delta(u')(u) \cdot R(u', u)),$
	\item An edge $(u', u_i)$ to every successor $u_i$ of $u$, such that $R(u', u_i) = R(u, u_i)$ and $\delta(u')(u_i) = \delta(u')(u) \cdot \delta(u)(u_i) \cdot \lambda.$
\end{compactitem}

This construction is shown in Figure~\ref{fig:vertex-rm-sum}. As shown above, using this construction the value of $y_v$ remains the same in solutions of $S_C$ and $S_{\overline{C}}.$ There are two special cases that can cause this construction to fail. However, we can avoid both of these cases using simple transformations in the graph before applying this construction. We now describe how we handle each of them:
\begin{compactitem}
	\item \emph{Parallel Edges.} If two edges with the same direction are created between the same pair $(u, v)$ of vertices, then we replace them with a single edge. If the $\delta$ values of initial edges were $\delta_1, \delta_2$ and their $R$ values were $r_1, r_2$, we set $\delta(u)(v) = \delta_1 + \delta_2$ and $R(u, v) = \frac{\delta_1 \cdot r_1 + \delta_2 \cdot r_2}{\delta_1 + \delta_2}.$ It is straigthforward to verify that this transformation is sound, i.e.~it does not change the solution of the corresponding system.  
	\item \emph{Self-loops.} If a self-loop $(u, u)$ appears in our graph, this is equivalent to having an equation $\equ_u := y_u = R$ in the linear system, in which $R$ is a linear expression that contains a non-zero multiple of $y_u.$ In this case, we simplify this equation to $y_u = R'$ by moving the summand containing $y_u$ to the left hand side and multiplying both sides by a suitable factor. We then update the outgoing edges of $u$ in our graph to model the new system. Note that this update does not add any new edges to the graph, except possibly the edge $(u, \one)$ for handling leftover constant factors.
\end{compactitem}

\begin{figure}[H]
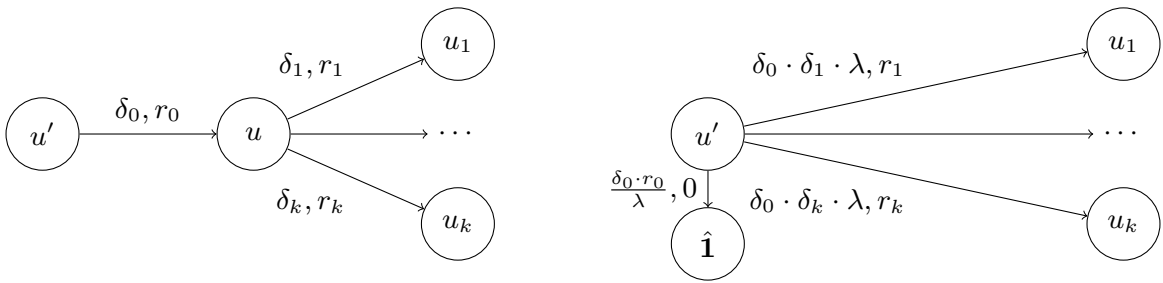

	\begin{center}
		\include{fig-vertex-rm-sum}
	\end{center}
	\caption{Removing $u$ from $C$ (left) to obtain $\overline{C}$ (right). The vertex $u'$ is a predecessor of $u$ and $u_1, \ldots, u_k$ are its successors. Each edge is labelled with its $\delta$ and $R$ values.}
	\label{fig:vertex-rm-sum}
\end{figure}

As in the previous section, we can solve the problem on the smaller $\overline{C}$ and then use the equation $\equ_u$ to compute the value of $y_u$ in the solution to $S_C$. This algorithm's runtime can be analyzed exactly as before. We have to remove $n$ vertices and each removal takes $O(n^2)$ for a total runtime of $O(n^3)$. To obtain a better algorithm that exploits tree decompositions, we can use the exact same removal order as in the previous section, leading to the same runtime, i.e.~$O(n \cdot t^2).$ Note that we have added $\one$ to the associated vertex set of every bag, so the tree decomposition always remains valid throughout our algorithm. Given this discussion, we have the following theorem:

\begin{theorem}
	Given an MC with $n$ vertices and treewidth $t$ and an optimal tree decomposition of the MC, the algorithm described in this section computes expected discounted sums from every vertex of the MC in $O(n \cdot t^2).$
\end{theorem}

\subsection{Solving Systems of Equations with Constant-Treewidth Primal Graphs} \label{sec:system}

The ideas used in the previous section can be extended to obtain faster algorithms for solving any linear system whose primal graph has a small treewidth. However, new subtleties arise, given that general linear systems might have no solution or infinitely many solutions. In contrast, the systems $S_C$ discussed in the previous section were guaranteed to have a unique solution. We consider a system $S$ of $m$ linear equations over $n$ real unknowns as input, and assume that its primal graph $G(S)$ has treewidth $t$. 

Our algorithm for solving $S$ is similar to our previous algorithms, and is actually what most students are taught in junior high school. We take an arbitrary unknown $x$ and choose an arbitrary equation $\equ$ in which $x$ appears with a non-zero coefficient. We then rewrite $\equ$ as $x = R_x$, where $R_x$ is a linear expression based on other unknowns. Finally, we replace every occurrence of $x$ in other equations with $R_x$ and solve the resulting smaller system $\overline{S}$. If $\overline{S}$ has no solutions or inifinitely many solutions, then so does $S$. Otherwise, we evaluate $R_x$ in the solution of $\overline{S}$ to get the solution value for $x$. Using this algorithm, we have to remove $O(n)$ unknowns. When removing $x$, we might have to replace an expression of size $O(n)$, i.e.~$R_x$, in $O(m)$ potential other equations where $x$ has appeared. Hence, the overall runtime is $O(n^2 \cdot m).$

Given a tree decomposition $(T, E_T)$ of the primal graph $G(S)$, we choose the unknows in the usual order, i.e.~we always choose an unknown $x$ that appears only in a leaf bag. If $x$ does not appear in any equations, then we can simply remove it and then $S$ is satisfiable iff $\overline{S}$ is satisfiable. Moreover, if $S$ is satisfiable, then it has infinitely many solutions, given that $x$ is not restricted. Otherwise, there is an equation $\equ$ in which $x$ appears with non-zero coefficient, and hence we can rewrite this equation as $x = R_x.$ Note that $x$ has $O(t)$ neighbors in $G(S)$, given that it only appears in a leaf bag and all of its neighbors should also appear in the same bag, hence the length of $R_x$ is $O(t),$ too. 

The problem is that $x$ might have appeared in any of the other $O(m)$ equations. Hence, replacing it with $R_x$ in every equation will lead to a runtime of $O(m \cdot t)$. We repeat this for every unknown, so our total runtime is $O(n \cdot m \cdot t),$ which is not linear. 

The crucial observation is that while $x$ might have appeared in as many as $m$ equations, not all of them are linearly independent. Let $\mathfrak{E}_x$ be the set of equations containing $x$ and $l$ be the leaf bag in which $x$ appears and assume that $V_l = \{x, y_1, \ldots, y_{k-1}\}.$ Then the only unknowns that can appear together with $x$ in an equation are $y_1, \ldots, y_{k-1}.$ In other words, all equations in $\mathfrak{E}_x$ are over $V_l.$ Hence, we can apply the Gram-Schmidt process on $\mathfrak{E}_x$ to remove the unnecessary equations and only keep at most $k$ equations that form an orthogonal basis (or alternatively realize that the system is unsatisfiable). Given that we are operating in dimension $k = O(t)$, this will take $O(t^2 \cdot \vert \mathfrak{E}_x \vert)$ time. See Appendix~\ref{app:pseudo} for a pseudocode. As in previous algorithms, our approach always keeps the tree decomposition valid. Moreover, as argued above, its runtime is $O((n+m) \cdot t^2),$ which is linear in the size of the system. Hence, we have the following theorem:

\begin{theorem} \label{thm:system}
	Given a system of $m$ linear equations over $n$ unknowns, its primal graph, and a tree decomposition of the primal graph with width $t$, our algorithm solves the system in time $O((n+m) \cdot t^2).$
\end{theorem}

The algorithm can easily be extended to find a basis for the solution set.

\subsection{Computing Expected Mean Payoffs in Constant Treewidth} \label{sec:mp}

\paragraph{Strongly Connected Components.} Given an MC $C=(V, E, \delta),$ a \emph{Strongly Connected Component} (SCC) is a maximal subset $A \subseteq V$, such that for every pair of vertices $u, v \in A$, there is a path from $u$ to $v$ in $C$. An SCC $B$ is called a \emph{Bottom Strongly Connected Component} (BSCC) if no other SCC is reachable from $B$. It is well-known that every vertex belongs to a unique SCC and that there is a linear-time algorithm that computes the SCCs and BSCCs of any given MC~\cite{leiserson2001introduction}. An MC is called \emph{ergodic} if its vertex set consists of only a single BSCC.

\paragraph{Limiting Distribution~\textnormal{\cite{norris1998markov}}.} Given an ergodic MC $C = (B, E, \delta)$ with a single BSCC $B$ and an arbitrary vertex $u \in B$, we define the \emph{limiting distribution} $\ld$ over $B$ as follows:
$
\ld(v) := \lim_{n \rightarrow \infty} \expv_u \left[ \frac{1}{n} \cdot \vert \{ i~\vert~0 \leq i < n \wedge \pi_i = v \} \vert \right],
$
where $\pi$ is a random walk beginning at $u$. Informally, $\ld(v)$ is the fraction of time that we are expected to spend in vertex $v$, when we start a random walk in $C$. Note that due to ergodicity, the starting vertex of the random walk does not matter. We can similarly define a limiting distribution $\ld^E$ over the edges of $C$ by letting $\ld^E(u, v) := \ld(u) \cdot \delta(u)(v)$.

From the definition above, it is easy to see that the mean payoff value $\expmeanpayoff(u)$ is the same for every vertex $u \in B$ of the ergodic MC. More specifically, we have $\expmeanpayoff(u) = \sum_{(v_1, v_2) \in E}  R(v_1, v_2) \cdot \ld^E(v_1, v_2).$ Therefore, computing the $\expmeanpayoff$ values is reduced to computing the limiting distribution.

Now consider a general MC $C = (V, E, \delta)$ and a vertex $u \in V.$ If $u$ is in a BSCC $B,$ then any path starting from $u$ will never leave $B.$ Therefore, $\expmeanpayoff_V(u) = \expmeanpayoff_B(u).$ On the other hand, if $u$ is in a non-bottom SCC $A,$ then the random walk beginning from $u$ will eventually reach a BSCC almost-surely (with probability $1$). Let $B_1, B_2, \ldots$ be the BSCCs of $C$ and $b_i \in B_i.$ Hence, given that we can ignore a finite prefix when computing mean payoffs, the expected mean payoff from $u$ is $$\expmeanpayoff(u) = \sum_{i} \hitprob(u, B_i) \cdot \expmeanpayoff(b_i) = \sum_{i} \hitprob(u, b_i) \cdot \expmeanpayoff(b_i).$$ 
Every vertex in $B_i$ has the same expected mean payoff and  will be reached from every other vertex in $B_i$ with probability $1$, i.e.~hitting probabilities between pairs of vertices in the same BSCC $B_i$ are always $1$, hence the choice of $b_i$ is arbitrary.

We use the two observations above to compute expected mean payoffs in a given MC $C$. Algorithm~\ref{algo:main-mp} summarizes our approach. Hence, the problem is reduced to computing $\ld$ (Line~5) and hitting probabilities (Lines~11--12). We now explain how we handle each of these two subproblems.

\SetAlFnt{\normalsize}
\begin{algorithm}
	\newcommand{\sol}{\textsl{solution}}
	\DontPrintSemicolon
	\SetKw{True}{true}
	\SetKw{nosol}{Unsatisfiable}
	\SetKw{infsol}{Underdetermined}
	\SetKw{Break}{break}
	\SetKwFunction{FMain}{ComputeExpMP}
	\SetKwFunction{GS}{Gramm-Schmidt}
	\SetKwBlock{Repeat}{repeat}{}
\SetKwProg{Fn}{Function}{:}{}
\Fn{\FMain{$C=(V, E, \delta)$}}{
		$B_1, B_2, \ldots \leftarrow $ BSCCs of $C$\;
		Choose an arbitrary $b_i$ from each $B_i$\;
		\ForEach{$B_i$}
		{
			Compute $\ld$ for $(B_i, E \cap (B_i \times B_i), \delta)$\;
			\ForEach{$(v_1, v_2) \in E \cap (B_i \times B_i)$}
			{
				$\ld^E(v_1, v_2) \leftarrow \ld(v_1) \cdot \delta(v_1)(v_2)$
			}
			$x \leftarrow \sum_{(v_1, v_2) \in E \cap (B_i \cdot B_i)}  R(v_1, v_2) \cdot \ld^E(v_1, v_2)$\;
			\ForEach{$u \in B_i$}
			{
				$\expmeanpayoff(u) \leftarrow x$\;
			}
		}
		\ForEach{$u \in V \setminus \bigcup B_i$}
		{
			$\expmeanpayoff(u) \leftarrow \sum_{i} \hitprob(u, b_i) \cdot \expmeanpayoff(b_i)$
		}
	}
	\vspace{2mm}
	\caption{Computing expected mean payoffs in a given MC $C$.}
	\label{algo:main-mp}
\end{algorithm}

\paragraph{Computing Limiting Distribution of an Ergodic MC.} Let $C = (B, E, \delta)$ be an ergodic MC. We define the linear system $S_C$ as follows:
\begin{compactitem}
	\item We add a variable $x_u$ for each vertex $u \in B$.
	\item For each vertex $u \in B$ with \emph{predecessors} $u_1, u_2, \ldots, u_k$, we add a constraint $x_u = \sum_{i=1}^k x_{u_i} \cdot \delta(u_i)(u).$
	\item We add the constraint $\sum_{u \in B} x_u = 1.$
\end{compactitem}
It is well-known that $S_C$ has a unique solution in which the value of each $x_u$ is equal to $\ld(u)$~\cite{norris1998markov}. Unfortunately, the last constraint includes all of the variables in the system and hence the primal graph of our system does not have constant treewidth. However, this is a minor restriction. We can consider the system $S'_C$ obtained by ignoring the last constraint. This system is homogeneous and its primal graph is the isomorphic to $(V, E)$ and has treewidth $t$. Hence, we can use the algorithm of Section~\ref{sec:system} to find an arbitrary solution to $S'_C.$ We can then scale all the values in our solution to satisfy the constraint $\sum_{u \in B} x_u = 1,$ hence obtaining the unique solution of $S_C.$ Therefore, Line~5 of Algorithm~\ref{algo:main-mp} takes $O(\vert B_i \vert \cdot t^2)$ time according to Theorem~\ref{thm:system}.

\paragraph{Computing Expected Mean Payoff for non-BSCC vertices.} We can compute all the values of $\expmeanpayoff(u)$ for $u \in V \setminus \bigcup B_i$ (Lines~11--12) with \emph{a single} call to our algorithm for hitting probabilities (Algorithm~\ref{algo:vertex-rem-td}, Section~\ref{sec:hitprobtd}). Note that Algorithm~\ref{algo:vertex-rem-td} does not rely on the premise that the function $\delta$ can only have values between $0$ and $1$. Hence, we can set all the $b_i$'s as targets, but when merging them to a single target $\target$, we set $\delta(b_i)(\target) = \expmeanpayoff(b_i),$ which was computed in Line~10. This ensures that the value computed for $\expmeanpayoff(u)$ is exactly the RHS of Line~11 in Algorithm~\ref{algo:main-mp}. Using this trick, the runtime of Lines~10--11 of our algorithm is $O(n \cdot t^2)$ as per Theorem~\ref{thm:hitting}. 

\begin{theorem} \label{thm:mp}
	Given an MC with $n$ vertices and treewidth $t$ and an optimal tree decomposition, Algorithm~\ref{algo:main-mp} computes expected mean payoffs from every vertex in $O(n \cdot t^2).$
\end{theorem}

\paragraph{Remark.}
	In SI over MDPs with mean payoff objectives, one also needs to compute additional values, called \emph{potentials} or \emph{biases}~\cite{Kretinsky17,puterman2014markov}. However, this computation is classically reduced to solving a system of linear equations whose primal graph is the MDP. Hence, the algorithm of Section~\ref{sec:system} can be applied, and our improvements for computing mean payoff in MCs extend to MDPs.

%% file: fig-vertex-rm.tex
\begin{tikzpicture}[scale=1.2, transform shape, node distance = {1cm and 5mm}, v/.style = {draw, circle, minimum size=8mm}, b/.style = {draw, rectangle}]

	\node (up) [v] {$u'$};
	\node (u) [v, right = of up] {$u$};
	\node (us) [v, right = of u] {$u''$};
	\node (up2) [v, below = of up] {$u'$};
	\node (u2) [right = of up2] {};
	\node (us2) [v, below = of us] {$u''$};
	
	\node (C) [left = of up]{$C$};
	\node (Cb) [left = of up2] {$\overline{C}$};
	
	\path (up) edge[->] node[above] {$p_1$} (u); 
	\path (u) edge[->] node[above] {$p_2$} (us); 
	\path (up2) edge[->] node[above] {$p_1 \cdot p_2$} (us2); 
	
	\node (rien) [right = 2cm of us] {};
	
	\node (up3) [v, right = of rien] {$u'$};
	\node (u3) [v, right = of up3] {$u$};
	\node (us3) [v, right = of u3] {$u''$};
	\node (up4) [v, below = of up3] {$u'$};
	\node (u4) [right = of up4] {};
	\node (us4) [v, below = of us3] {$u''$};
	
	\path (up3) edge[->] node[above] {$p_1$} (u3); 
	\path (u3) edge[->] node[above] {$p_2$} (us3); 
	\path (up4) edge[->] node[above] {$p_1 \cdot p_2 + p_3$} (us4);
	\path (up3) edge[->, bend right, looseness=2] node[below] {$p_3$} (us3); 
	
\end{tikzpicture}

%% file: fig-long-example-1.tex
\begin{tikzpicture}[scale=0.68, transform shape, node distance = {5mm and 7mm}, v/.style = {draw, circle}, vr/.style = {draw, circle, color=red},b/.style = {draw, rectangle}, br/.style = {draw, rectangle, color=red}, bb/.style = {draw, rectangle, color=blue}]
	
	\node (1) [v] {1};
	\node (2) [v, below = of 1] {2};
	\node (3) [v, right = of 2] {3};
	\node (4)  [v, right = of 3] {4};
	\node (5)  [v, below = of 2] {5};
	\node (6)  [v, right = of 5, color = cadmiumgreen] {6};
	\node (7)  [v, right = of 6] {7};
	\node (8)  [v, below = of 6] {8};
	\node (9)  [vr, right = of 8] {9};
	
	\draw (1) [->] to (2);
	\draw (1) [<-] to (3);
	\draw (2) [->] to (3);
	\draw (2) [->] to (5);
	\draw (2) [->] to (6);
	\draw (3) [->, bend right] to (4);
	\draw (4) [->, bend right] to (3);
	\draw (3) [<-] to (6);
	\draw (4) [->] to (6);
	\draw (4) [->] to (7);
	\draw (5) [->] to (6);
	\draw (5) [<-] to (8);
	\draw (6) [<-] to (7);
	\draw (6) [->] to (8);
	\draw (7) [->] to (9);
	\path (9) edge [->, out=330,in=30,looseness=5] (9);
	
	\node (B) [b, right = 4.8cm of 1] {$\{2, 3, 6\}$};
	\node (D) [b, below = of B] {$\{3, 4, 6\}$};
	\node (C) [b, right = of D] {$\{2, 5, 6\}$};
	\node (A) [b, left = of D] {$\{1, 2, 3\}$};
	\node (E) [b, below = of C] {$\{5, 6, 8\}$};
	\node (F) [b, below = of D] {$\{4, 6, 7\}$};
	\node (G) [bb, below = of F] {$\{7, 9\}$};
	
	\draw (A) to (B);
	\draw (B) to (C);
	\draw (B) to (D);
	\draw (C) to (E);
	\draw (B) to (D);
	\draw (D) to (F);
	\draw (F) to (G);
	
	\node (Space) [above = 2mm of 1] {};
	
\end{tikzpicture}

%% file: fig-long-example-2.tex
\begin{tikzpicture}[scale=0.68, transform shape, node distance = {5mm and 7mm}, v/.style = {draw, circle}, vr/.style = {draw, circle, color=red},b/.style = {draw, rectangle}, br/.style = {draw, rectangle, color=red}, bb/.style = {draw, rectangle, color=blue}]
	
	\node (1) [v] {1};
	\node (2) [v, below = of 1] {2};
	\node (3) [v, right = of 2] {3};
	\node (4)  [v, right = of 3] {4};
	\node (5)  [v, below = of 2] {5};
	\node (6)  [v, right = of 5, cadmiumgreen] {6};
	\node (7)  [v, right = of 6] {7};
	\node (8)  [v, below = of 6] {8};
	
	\draw (1) [->] to (2);
	\draw (1) [<-] to (3);
	\draw (2) [->] to (3);
	\draw (2) [->] to (5);
	\draw (2) [->] to (6);
	\draw (3) [->, bend right] to (4);
	\draw (4) [->, bend right] to (3);
	\draw (3) [<-] to (6);
	\draw (4) [->] to (6);
	\draw (4) [->] to (7);
	\draw (5) [->] to (6);
	\draw (5) [<-] to (8);
	\draw (6) [<-] to (7);
	\draw (6) [->] to (8);

	\node (B) [b, right = 4.8cm of 1] {$\{2, 3, 6\}$};
	\node (D) [b, below = of B] {$\{3, 4, 6\}$};
	\node (C) [b, right = of D] {$\{2, 5, 6\}$};
	\node (A) [b, left = of D] {$\{1, 2, 3\}$};
	\node (E) [b, below = of C] {$\{5, 6, 8\}$};
	\node (F) [b, below = of D] {$\{4, 6, 7\}$};
	\node (G) [br, below = of F] {$\{7\}$};
	
	\draw (A) to (B);
	\draw (B) to (C);
	\draw (B) to (D);
	\draw (C) to (E);
	\draw (B) to (D);
	\draw (D) to (F);
	\draw (F) to (G);
	
		\node (Space) [above = 2mm of 1] {};
	
\end{tikzpicture}

%% file: fig-long-example-3.tex
\begin{tikzpicture}[scale=0.68, transform shape, node distance = {5mm and 7mm}, v/.style = {draw, circle}, vr/.style = {draw, circle, color=red},b/.style = {draw, rectangle}, br/.style = {draw, rectangle, color=red}, bb/.style = {draw, rectangle, color=blue}]
	
	\node (1) [v] {1};
	\node (2) [v, below = of 1] {2};
	\node (3) [v, right = of 2] {3};
	\node (4)  [v, right = of 3] {4};
	\node (5)  [v, below = of 2] {5};
	\node (6)  [v, right = of 5, cadmiumgreen] {6};
	\node (7)  [v, right = of 6] {7};
	\node (8)  [vr, below = of 6] {8};
	
	\draw (1) [->] to (2);
	\draw (1) [<-] to (3);
	\draw (2) [->] to (3);
	\draw (2) [->] to (5);
	\draw (2) [->] to (6);
	\draw (3) [->, bend right] to (4);
	\draw (4) [->, bend right] to (3);
	\draw (3) [<-] to (6);
	\draw (4) [->] to (6);
	\draw (4) [->] to (7);
	\draw (5) [->] to (6);
	\draw (5) [<-] to (8);
	\draw (6) [<-] to (7);
	\draw (6) [->] to (8);

	\node (B) [b, right = 4.8cm of 1] {$\{2, 3, 6\}$};
	\node (D) [b, below = of B] {$\{3, 4, 6\}$};
	\node (C) [b, right = of D] {$\{2, 5, 6\}$};
	\node (A) [b, left = of D] {$\{1, 2, 3\}$};
	\node (E) [bb, below = of C] {$\{5, 6, 8\}$};
	\node (F) [b, below = of D] {$\{4, 6, 7\}$};
	
	\draw (A) to (B);
	\draw (B) to (C);
	\draw (B) to (D);
	\draw (C) to (E);
	\draw (B) to (D);
	\draw (D) to (F);
	
		\node (Space) [above = 2mm of 1] {};
	
\end{tikzpicture}

%% file: fig-long-example-4.tex
\begin{tikzpicture}[scale=0.68, transform shape, node distance = {5mm and 7mm}, v/.style = {draw, circle}, vr/.style = {draw, circle, color=red},b/.style = {draw, rectangle}, br/.style = {draw, rectangle, color=red}, bb/.style = {draw, rectangle, color=blue}]
	
	\node (1) [v] {1};
	\node (2) [v, below = of 1] {2};
	\node (3) [v, right = of 2] {3};
	\node (4)  [v, right = of 3] {4};
	\node (5)  [v, below = of 2] {5};
	\node (6)  [v, right = of 5, cadmiumgreen] {6};
	\node (7)  [v, right = of 6] {7};
	
	\draw (1) [->] to (2);
	\draw (1) [<-] to (3);
	\draw (2) [->] to (3);
	\draw (2) [->] to (5);
	\draw (2) [->] to (6);
	\draw (3) [->, bend right] to (4);
	\draw (4) [->, bend right] to (3);
	\draw (3) [<-] to (6);
	\draw (4) [->] to (6);
	\draw (4) [->] to (7);
	\draw (5) [->, bend right] to (6);
	\draw (6) [->, bend right] to (5);
	\draw (6) [<-] to (7);

	\node (B) [b, right = 4.8cm of 1] {$\{2, 3, 6\}$};
	\node (D) [b, below = of B] {$\{3, 4, 6\}$};
	\node (C) [b, right = of D] {$\{2, 5, 6\}$};
	\node (A) [b, left = of D] {$\{1, 2, 3\}$};
	\node (E) [br, below = of C] {$\{5, 6\}$};
	\node (F) [b, below = of D] {$\{4, 6, 7\}$};
	\node (nn) [v, below = of 6, white]{$1$};
	
	\draw (A) to (B);
	\draw (B) to (C);
	\draw (B) to (D);
	\draw (C) to (E);
	\draw (B) to (D);
	\draw (D) to (F);
	
	\node (Space) [above = 2mm of 1] {};
	
\end{tikzpicture}

%% file: fig-long-example-5.tex
\begin{tikzpicture}[scale=0.68, transform shape, node distance = {5mm and 7mm}, v/.style = {draw, circle}, vr/.style = {draw, circle, color=red},b/.style = {draw, rectangle}, br/.style = {draw, rectangle, color=red}, bb/.style = {draw, rectangle, color=blue}]
	
	\node (1) [v] {1};
	\node (2) [v, below = of 1] {2};
	\node (3) [v, right = of 2] {3};
	\node (4)  [v, right = of 3] {4};
	\node (5)  [vr, below = of 2] {5};
	\node (6)  [v, right = of 5, cadmiumgreen] {6};
	\node (7)  [v, right = of 6] {7};
	
	\draw (1) [->] to (2);
	\draw (1) [<-] to (3);
	\draw (2) [->] to (3);
	\draw (2) [->] to (5);
	\draw (2) [->] to (6);
	\draw (3) [->, bend right] to (4);
	\draw (4) [->, bend right] to (3);
	\draw (3) [<-] to (6);
	\draw (4) [->] to (6);
	\draw (4) [->] to (7);
	\draw (5) [->, bend right] to (6);
	\draw (6) [->, bend right] to (5);
	\draw (6) [<-] to (7);

	\node (B) [b, right = 4.8cm of 1] {$\{2, 3, 6\}$};
	\node (D) [b, below = of B] {$\{3, 4, 6\}$};
	\node (C) [bb, right = of D] {$\{2, 5, 6\}$};
	\node (A) [b, left = of D] {$\{1, 2, 3\}$};
	\node (F) [b, below = of D] {$\{4, 6, 7\}$};
	
	\draw (A) to (B);
	\draw (B) to (C);
	\draw (B) to (D);
	\draw (B) to (D);
	\draw (D) to (F);
	
	\node (Space) [above = 2mm of 1] {};
	
\end{tikzpicture}

%% file: fig-long-example-6.tex
\begin{tikzpicture}[scale=0.68, transform shape, node distance = {5mm and 7mm}, v/.style = {draw, circle}, vr/.style = {draw, circle, color=red},b/.style = {draw, rectangle}, br/.style = {draw, rectangle, color=red}, bb/.style = {draw, rectangle, color=blue}]
	
	\node (1) [v] {1};
	\node (2) [v, below = of 1] {2};
	\node (3) [v, right = of 2] {3};
	\node (4)  [v, right = of 3] {4};
	\node (6)  [v, right = of 5, cadmiumgreen] {6};
	\node (7)  [v, right = of 6] {7};
	
	\draw (1) [->] to (2);
	\draw (1) [<-] to (3);
	\draw (2) [->] to (3);
	\draw (2) [->] to (6);
	\draw (3) [->, bend right] to (4);
	\draw (4) [->, bend right] to (3);
	\draw (3) [<-] to (6);
	\draw (4) [->] to (6);
	\draw (4) [->] to (7);
	\draw (6) [<-] to (7);
	\path (6) edge [->, out=180,in=220,looseness=5] (6);

	\node (B) [b, right = 4.8cm of 1] {$\{2, 3, 6\}$};
	\node (D) [b, below = of B] {$\{3, 4, 6\}$};
	\node (C) [br, right = of D] {$\{2, 6\}$};
	\node (A) [b, left = of D] {$\{1, 2, 3\}$};
	\node (F) [b, below = of D] {$\{4, 6, 7\}$};
	
	\draw (A) to (B);
	\draw (B) to (C);
	\draw (B) to (D);
	\draw (B) to (D);
	\draw (D) to (F);
	
	\node (Space) [above = 2mm of 1] {};
	
\end{tikzpicture}

%% file: fig-long-example-7.tex
\begin{tikzpicture}[scale=0.68, transform shape, node distance = {5mm and 7mm}, v/.style = {draw, circle}, vr/.style = {draw, circle, color=red},b/.style = {draw, rectangle}, br/.style = {draw, rectangle, color=red}, bb/.style = {draw, rectangle, color=blue}]
	
	\node (1) [v] {1};
	\node (2) [v, below = of 1] {2};
	\node (3) [v, right = of 2] {3};
	\node (4)  [v, right = of 3] {4};
	\node (6)  [v, right = of 5, cadmiumgreen] {6};
	\node (7)  [vr, right = of 6] {7};
	
	\draw (1) [->] to (2);
	\draw (1) [<-] to (3);
	\draw (2) [->] to (3);
	\draw (2) [->] to (6);
	\draw (3) [->, bend right] to (4);
	\draw (4) [->, bend right] to (3);
	\draw (3) [<-] to (6);
	\draw (4) [->] to (6);
	\draw (4) [->] to (7);
	\draw (6) [<-] to (7);
	\path (6) edge [->, out=180,in=220,looseness=5] (6);

	\node (B) [b, right = 6cm of 1] {$\{2, 3, 6\}$};
	\node (D) [b, below = of B] {$\{3, 4, 6\}$};
	\node (A) [b, left = of D] {$\{1, 2, 3\}$};
	\node (F) [bb, below = of D] {$\{4, 6, 7\}$};
	
	\draw (A) to (B);
	\draw (B) to (D);
	\draw (B) to (D);
	\draw (D) to (F);
	
	\node (Space) [above = 2mm of 1] {};
	
\end{tikzpicture}

%% file: fig-long-example-8.tex
\begin{tikzpicture}[scale=0.68, transform shape, node distance = {5mm and 7mm}, v/.style = {draw, circle}, vr/.style = {draw, circle, color=red},b/.style = {draw, rectangle}, br/.style = {draw, rectangle, color=red}, bb/.style = {draw, rectangle, color=blue}]
	
	\node (1) [v] {1};
	\node (2) [v, below = of 1] {2};
	\node (3) [v, right = of 2] {3};
	\node (4)  [v, right = of 3] {4};
	\node (6)  [v, right = of 5, cadmiumgreen] {6};
	
	\draw (1) [->] to (2);
	\draw (1) [<-] to (3);
	\draw (2) [->] to (3);
	\draw (2) [->] to (6);
	\draw (3) [->, bend right] to (4);
	\draw (4) [->, bend right] to (3);
	\draw (3) [<-] to (6);
	\draw (4) [->] to (6);
	\path (6) edge [->, out=180,in=220,looseness=5] (6);

	\node (B) [b, right = 6cm of 1] {$\{2, 3, 6\}$};
	\node (D) [b, below = of B] {$\{3, 4, 6\}$};
	\node (A) [b, left = of D] {$\{1, 2, 3\}$};
	\node (F) [br, below = of D] {$\{4, 6\}$};
	
	\draw (A) to (B);
	\draw (B) to (D);
	\draw (B) to (D);
	\draw (D) to (F);
	
	\node (Space) [above = 2mm of 1] {};
	
\end{tikzpicture}

%% file: fig-long-example-9.tex
\begin{tikzpicture}[scale=0.68, transform shape, node distance = {5mm and 7mm}, v/.style = {draw, circle}, vr/.style = {draw, circle, color=red},b/.style = {draw, rectangle}, br/.style = {draw, rectangle, color=red}, bb/.style = {draw, rectangle, color=blue}]
	
	\node (1) [v] {1};
	\node (2) [v, below = of 1] {2};
	\node (3) [v, right = of 2] {3};
	\node (4)  [vr, right = of 3] {4};
	\node (6)  [v, right = of 5, cadmiumgreen] {6};
	
	\draw (1) [->] to (2);
	\draw (1) [<-] to (3);
	\draw (2) [->] to (3);
	\draw (2) [->] to (6);
	\draw (3) [->, bend right] to (4);
	\draw (4) [->, bend right] to (3);
	\draw (3) [<-] to (6);
	\draw (4) [->] to (6);
	\path (6) edge [->, out=180,in=220,looseness=5] (6);

	\node (B) [b, right = 6cm of 1] {$\{2, 3, 6\}$};
	\node (D) [bb, below = of B] {$\{3, 4, 6\}$};
	\node (A) [b, left = of D] {$\{1, 2, 3\}$};
	
	\draw (A) to (B);
	\draw (B) to (D);
	\draw (B) to (D);
	
	\node (Space) [above = 2mm of 1] {};
	
\end{tikzpicture}

%% file: fig-long-example-10.tex
\begin{tikzpicture}[scale=0.68, transform shape, node distance = {5mm and 7mm}, v/.style = {draw, circle}, vr/.style = {draw, circle, color=red},b/.style = {draw, rectangle}, br/.style = {draw, rectangle, color=red}, bb/.style = {draw, rectangle, color=blue}]
	
	\node (1) [v] {1};
	\node (2) [v, below = of 1] {2};
	\node (3) [v, right = of 2] {3};
	\node (6)  [v, right = of 5, cadmiumgreen] {6};
	
	\draw (1) [->] to (2);
	\draw (1) [<-] to (3);
	\draw (2) [->] to (3);
	\draw (2) [->] to (6);
	\draw (3) [->, bend right] to (6);
	\draw (6) [->, bend right] to (3);
	\path (6) edge [->, out=180,in=220,looseness=5] (6);
	\path (3) edge [->, out=0,in=40,looseness=5] (3);

	\node (B) [b, right = 6cm of 1] {$\{2, 3, 6\}$};
	\node (D) [br, below = of B] {$\{3, 6\}$};
	\node (A) [b, left = of D] {$\{1, 2, 3\}$};
	
	\draw (A) to (B);
	\draw (B) to (D);
	\draw (B) to (D);
	
	\node (Space) [above = 2mm of 1] {};
	
\end{tikzpicture}

%% file: fig-long-example-11.tex
\begin{tikzpicture}[scale=0.68, transform shape, node distance = {5mm and 7mm}, v/.style = {draw, circle}, vr/.style = {draw, circle, color=red},b/.style = {draw, rectangle}, br/.style = {draw, rectangle, color=red}, bb/.style = {draw, rectangle, color=blue}]
	
	\node (1) [vr] {1};
	\node (2) [v, below = of 1] {2};
	\node (3) [v, right = of 2] {3};
	\node (6)  [v, right = of 5, cadmiumgreen] {6};
	
	\draw (1) [->] to (2);
	\draw (1) [<-] to (3);
	\draw (2) [->] to (3);
	\draw (2) [->] to (6);
	\draw (3) [->, bend right] to (6);
	\draw (6) [->, bend right] to (3);
	\path (6) edge [->, out=180,in=220,looseness=5] (6);
	\path (3) edge [->, out=0,in=40,looseness=5] (3);

	\node (B) [b, right = 6cm of 1] {$\{2, 3, 6\}$};
	\node (A) [bb, left = of D] {$\{1, 2, 3\}$};
	
	\draw (A) to (B);
	
	\node (Space) [above = 2mm of 1] {};
	
\end{tikzpicture}

%% file: fig-long-example-12.tex
\begin{tikzpicture}[scale=0.68, transform shape, node distance = {5mm and 7mm}, v/.style = {draw, circle}, vr/.style = {draw, circle, color=red},b/.style = {draw, rectangle}, br/.style = {draw, rectangle, color=red}, bb/.style = {draw, rectangle, color=blue}]
	
	\node (2) [v] {2};
	\node (3) [v, right = of 2] {3};
	\node (6)  [v, below = of 3, cadmiumgreen] {6};
	
	\draw (2) [->, bend right] to (3);
	\draw (3) [->, bend right] to (2);
	\draw (2) [->] to (6);
	\draw (3) [->, bend right] to (6);
	\draw (6) [->, bend right] to (3);
	\path (6) edge [->, out=180,in=220,looseness=5] (6);
	\path (3) edge [->, out=0,in=40,looseness=5] (3);

	\node (B) [b, right = 6cm of 1] {$\{2, 3, 6\}$};
	\node (A) [br, left = of D] {$\{2, 3\}$};
	
	\draw (A) to (B);
	
	\node (Space) [above = 2mm of 2] {};
	
\end{tikzpicture}

%% file: fig-long-example-13.tex
\begin{tikzpicture}[scale=0.68, transform shape, node distance = {5mm and 7mm}, v/.style = {draw, circle}, vr/.style = {draw, circle, color=red},b/.style = {draw, rectangle}, br/.style = {draw, rectangle, color=red}, bb/.style = {draw, rectangle, color=blue}]
	
	\node (2) [v] {2};
	\node (3) [vr, right = of 2] {3};
	\node (6)  [v, below = of 3, cadmiumgreen] {6};
	
	\draw (2) [->, bend right] to (3);
	\draw (3) [->, bend right] to (2);
	\draw (2) [->] to (6);
	\draw (3) [->, bend right] to (6);
	\draw (6) [->, bend right] to (3);
	\path (6) edge [->, out=180,in=220,looseness=5] (6);
	\path (3) edge [->, out=0,in=40,looseness=5] (3);

	\node (B) [bb, right = 6cm of 1] {$\{2, 3, 6\}$};
	
	\node (Space) [above = 2mm of 2] {};
	
\end{tikzpicture}

%% file: fig-long-example-14.tex
\begin{tikzpicture}[scale=0.68, transform shape, node distance = {5mm and 7mm}, v/.style = {draw, circle}, vr/.style = {draw, circle, color=red},b/.style = {draw, rectangle}, br/.style = {draw, rectangle, color=red}, bb/.style = {draw, rectangle, color=blue}]
	
	\node (2) [vr] {2};
	\node (6)  [v, below = of 3, cadmiumgreen] {6};
	
	\draw (2) [->, bend left] to (6);
	\draw (6) [->, bend left] to (2);
	\path (6) edge [->, out=180,in=220,looseness=5] (6);
	\path (2) edge [->, out=180,in=220,looseness=5] (2);

	\node (B) [bb, right = 6cm of 1] {$\{2, 3, 6\}$};
	
	\node (Space) [above = 2mm of 2] {};
	
\end{tikzpicture}

%% file: fig-one.tex
\begin{tikzpicture}[scale=0.6, transform shape, node distance = {5mm and 7mm}, v/.style = {draw, circle}, vr/.style = {draw, circle, color=red},b/.style = {draw, rectangle}, br/.style = {draw, rectangle, color=red}, bb/.style = {draw, rectangle, color=blue}]
	
	\node (1) [v] {\huge{$\one$}};

	\path (1) edge [->, out=180,in=220,looseness=5] node[left]{\Large{$1, 1-\lambda$}} (1);
	
\end{tikzpicture}

%% file: fig-vertex-rm-sum.tex
\begin{tikzpicture}[scale=1.2, transform shape, node distance = {4mm and 15mm}, v/.style = {draw, circle, minimum size=8mm}, b/.style = {draw, rectangle}]

	\node (up) [v] {$u'$};
	\node (u) [v, right = of up] {$u$};
	\node (dots) [right = of u] {$\cdots$};
	\node (u1) [v, above = of dots] {$u_1$};
	\node (uk) [v, below = of dots] {$u_k$};

	\node (up2) [v, right = 2cm of dots] {$u'$};
	\node (u2) [v, right = of up2, white] {$u$};
	\node (dots2) [right = of u2] {$\cdots$};
	\node (u12) [v, above = of dots2] {$u_1$};
	\node (uk2) [v, below = of dots2] {$u_k$};
	\node (one) [v, below = of up2] {$\one$};

	\path (up) edge[->] node[above] {$\delta_0, r_0$} (u); 
	\path (u) edge[->] node[above left] {$\delta_1, r_1$} (u1); 
	\path (u) edge[->] node[below left] {$\delta_k, r_k$} (uk); 
	\path (u) edge[->] (dots); 
	
	\path (up2) edge[->] node[above left] {$\delta_0 \cdot \delta_1 \cdot \lambda, r_1$} (u12); 
	\path (up2) edge[->] node[below left] {$\delta_0 \cdot \delta_k \cdot \lambda, r_k$} (uk2); 
	\path (up2) edge[->] (dots2); 
	\path (up2) edge[->] node[left] {$\frac{\delta_0 \cdot r_0}{\lambda}, 0$} (one);

\end{tikzpicture}

%% file: experiments.tex
\section{Experimental Results} \label{sec:exp}

In this section, we report on a C/C++ implementation of our algorithms  and provide a performance comparison with previous approaches in the literature.

\paragraph{Compared Approaches.} We consider the hitting probability and discounted sum problems for MCs and MDPs.
In the case of MCs, we directly use our algorithms from Section~\ref{sec:hitprobtd} and Section~\ref{subsec:discounted}.
For MDPs, we use strategy iteration, where we use the above algorithms for the strategy evaluation step in each iteration.
We compare our approach with the following alternatives:
\begin{itemize}
	\item \emph{Classical Approaches.} In case of MCs, we compare against an implementation of Gaussian elimination (\textsf{Gauss}) taken from~\cite{gauss}. For MDPs, we consider our own implementation of  value iteration (\textsf{VI}) and strategy iteration (\textsf{SI}).
	\item \emph{Numerical and Industrial Optimizers.} We
	 use \textsf{Matlab}~\cite{MATLAB10} and \textsf{Gurobi}~\cite{gurobi} to solve systems of linear equalities corresponding to MCs. For MDPs, we use \textsf{Matlab}~\cite{MATLAB10}, \textsf{Gurobi}~\cite{gurobi} and \textsf{lpsolve}~\cite{lpsolve} to handle the corresponding LPs.
	
	\item \emph{Probabilistic Model Checkers.} The well-known model checkers \textsf{Storm}~\cite{storm} and \textsf{Prism}~\cite{Kwiatkowska11} have standard procedures for computing hitting probabilities, but not for discounted sums. We therefore compare our runtimes on hitting probability instances with their runtimes. 
\end{itemize}

\smallskip

Despite the fact that treewidth has been extensively studied in verification and model checking~\cite{Obsrzalek03,Ferrara05}, including for the analysis of MDPs~\cite{Chatterjee13}, to the best of our knowledge there are no benchmark suites consisting of low-treewidth MCs/MDPs. Previous works such as~\cite{Chatterjee13} do not provide any experimental results.

\paragraph{Motivation for Benchmarks.}
The main motivation to study MCs/MDPs with small treewidth is that they occur naturally 
in static program analysis, where 
a key algorithmic problem is reachability on the CFGs,
e.g.~data-flow analyses in frameworks such as IFDS are reduced to reachability~\cite{ifds,ifds2}.
Moreover, probability annotations of the CFG are useful in many contexts such as (i)~in probabilistic programs where the branches are probabilistic~\cite{rajamani};
or
(ii)~when branch-profiling information is available that assigns probabilities to branch execution~\cite{ball1993branch,smith1998study}.
If we consider CFGs where all branches are deterministic or probabilistic, then we have MCs; 
and if there are also non-deterministic branches, then we have MDPs.
In both cases, the reachability analysis in CFGs with probability annotation corresponds to the computation of
hitting probabilities. Therefore, hitting probabilities can be used to answer questions like ``given the branch profiles, compute the probability that a given pointer is null in some instruction''.
Additionally, \cite{de2003discounting} shows how discounted-sum objectives are relevant in the analysis of systems,
e.g.~with discounted-sum reachability we can model that a later bug is better than an earlier one.
It is well-established that structured programs have small treewidth, both theoretically~\cite{thorup1998all} and 
experimentally~\cite{gustedt2002treewidth,Krause19,Burgstaller04,chatterjee2019treewidth}.
Thus, quantitative analysis of MCs/MDPs with small treewdith is a relevant problem 
in program analysis, and we consider benchmarks from this domain.

\paragraph{Benchmarks.} Given the points above, we used CFGs of the $40$ Java programs from the \textsf{DaCapo} suite~\cite{dacapo} as our benchmarks. They come in a variety of sizes, having between $33$ and $103918$ vertices and transitions. To obtain MDPs, we randomly (with probability $1/2$) turned each vertex into either a Player~$1$ vertex or a probabilistic one. Moreover, we assigned random probabilities to each outgoing edge of a probabilistic vertex. To obtain MCs, we did the same, except that we marked all vertices as probabilistic. For the hitting probabilities problem, we chose one random vertex from each connected component of the control flow graphs as a target. In case of discounted sum, we uniformly chose a discount factor between $0$ and $1$ for each instance, and also assigned random integral rewards between $-1000$ to $1000$ to each edge. Finally, we used \textsf{JTDec}~\cite{jtdec} to compute tree decompositions for our instances. In each case the width of the obtained decomposition was no more than $9$. See Appendix~\ref{app:exp} for a detailed overview of the 40 benchmarks used in our experimental results.

\paragraph{Results.} 
The runtimes are shown in Figures~\ref{fig:mc-hp}--\ref{fig:mdp-ds}. In each case, the benchmarks are sorted by their size. Note that the $y$-axes in these figures are in a \emph{logarithmic scale}.  For example, Figure~\ref{fig:mc-hp} corresponds to our experimental results for computing hitting probabilities in MCs. In this case, \textsf{Prism} is the slowest tool by far. On the other side of the spectrum, our approach beats every other method by one or more orders of magnitude. The gap is more apparent in case of MDPs (Figures~\ref{fig:mdp-hp}--\ref{fig:mdp-ds}).
Overall, we see that the new algorithms introduced in this work consistently outperform both existing practical approaches like \textsf{VI} and \textsf{SI}, and highly optimized solvers and model checkers like \textsf{Gurobi}, \textsf{Prism} and \textsf{Storm}, by one or more orders of magnitude.
Hence, the theoretical improvements are also realized in practice. See Appendix~\ref{app:exp} for detailed tables containing raw numbers.

\begin{figure}[H]
		\hspace{-1cm}
		\includegraphics[keepaspectratio,width=1.1\linewidth]{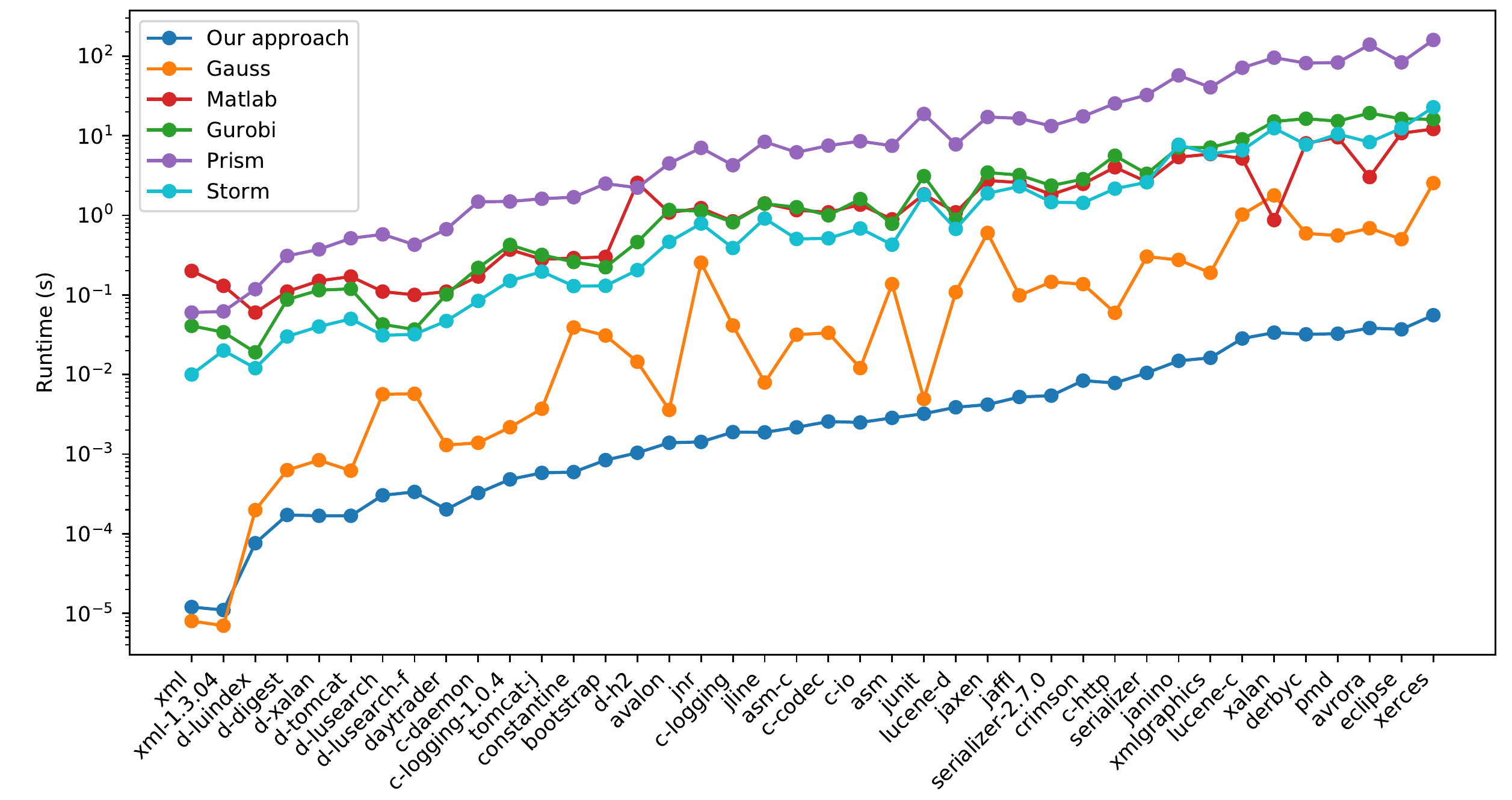}
\caption{Experimental Results for Computing Hitting Probabilities in MCs.}
\label{fig:mc-hp}
\end{figure}

\begin{figure}[H]
		\hspace{-1cm}
		\includegraphics[keepaspectratio,width=1.1\linewidth]{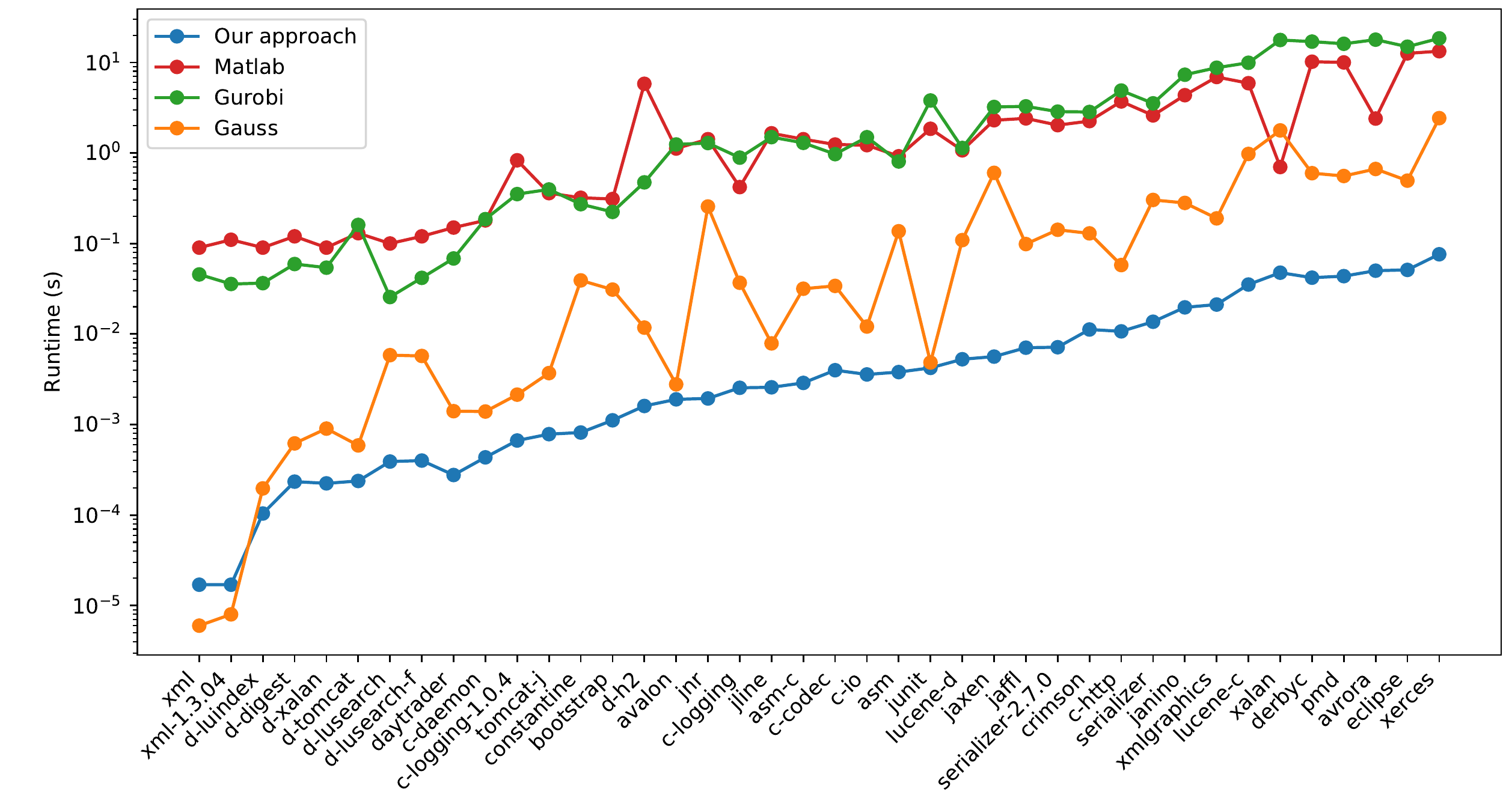}
\caption{Experimental Results for Computing Expected Discounted Sums in MCs.}
	\label{fig:mc-ds}
\end{figure}

\begin{figure}[H]
		\hspace{-1cm}
		\includegraphics[keepaspectratio,width=1.1\linewidth]{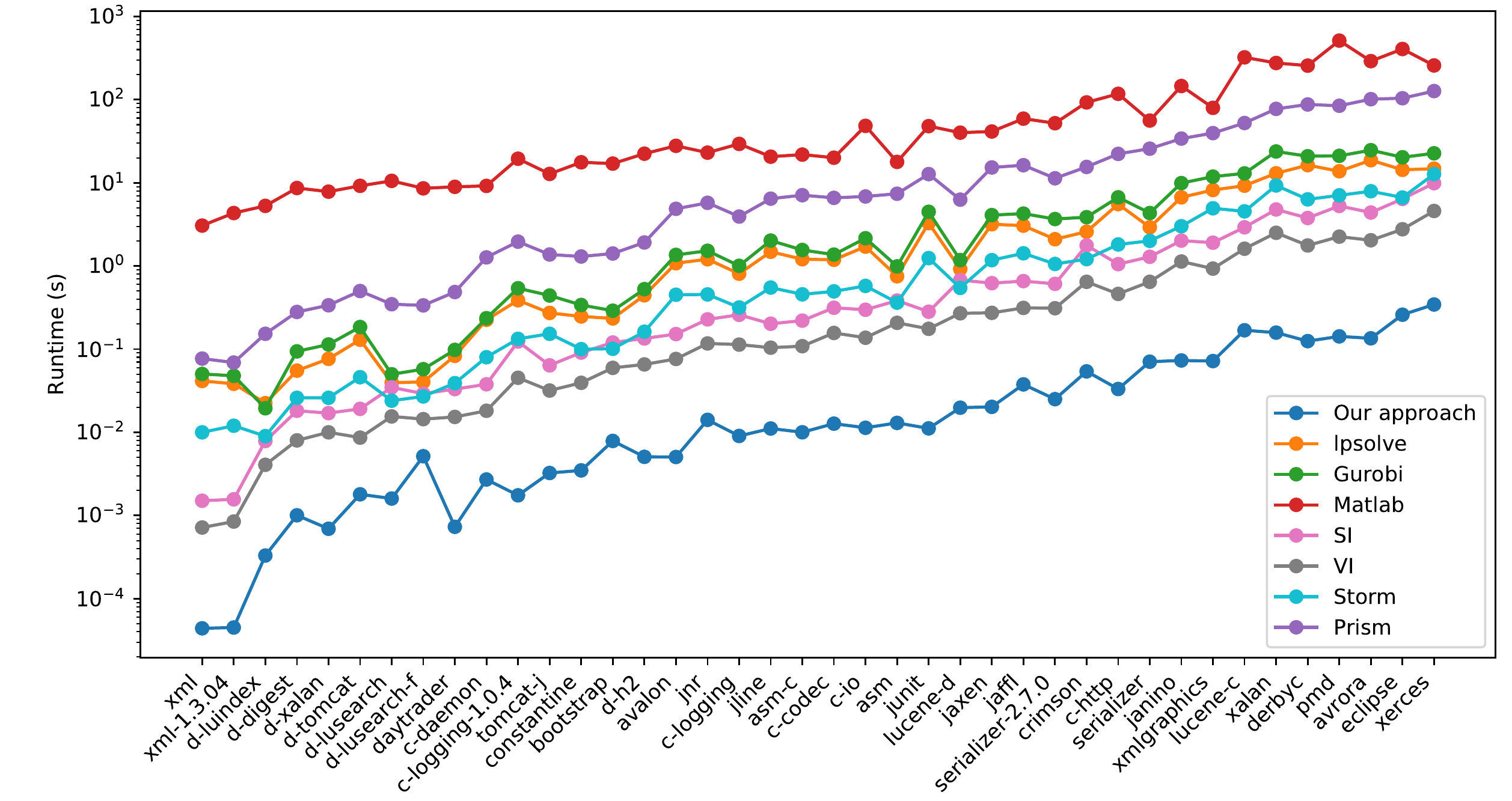}
	\caption{Experimental Results for Computing Hitting Probabilities in MDPs.}
	\label{fig:mdp-hp}
\end{figure}

\begin{figure}[H]
		\hspace{-1cm}
		\includegraphics[keepaspectratio,width=1.1\linewidth]{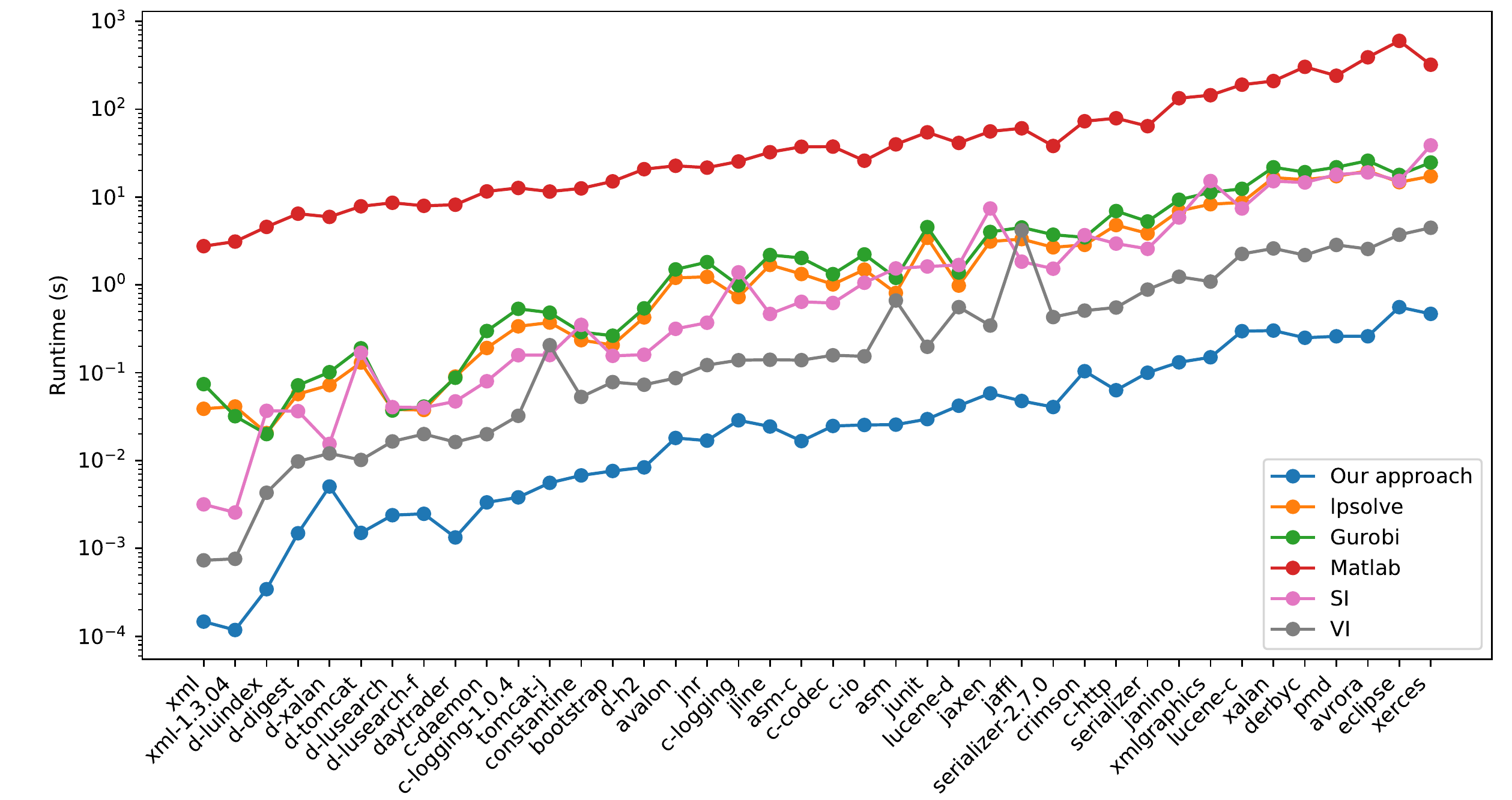}
	\caption{Experimental Results for Computing Expected Discounted Sums in MDPs.}
	\label{fig:mdp-ds}
\end{figure}

%% file: appendix.tex
\newpage
\appendix
\SetAlFnt{\footnotesize}
\section{Pseudocodes} \label{app:pseudo}

\begin{algorithm}
	\DontPrintSemicolon
	\SetKwFunction{FMain}{ComputeHitProbs}
	\SetKwProg{Fn}{Function}{:}{}
	\Fn{\FMain{$C = (V, E, \delta), \target$}}{
		\eIf{$V = \{\target\}$}
		{
			$\hitprob(\target, \target) \leftarrow 1$\;
		}
		{
			Choose an arbitrary $u \in V \setminus \{ \target \}$\;
			
			\eIf{$\delta(u)(u)  = 1$}{
				$\hitprob(u, \target) \leftarrow 0$\;
				\FMain($(V \setminus \{u\}, E, \delta), \target$)\;
			}
			{
				
				$f \leftarrow \frac{1}{1 - \delta(u)(u)}$\;
				$\delta(u)(u) \leftarrow 0$\;
				$E \leftarrow E \setminus \{(u, u)\}$
				
				\ForEach{$u'' \in V: (u, u'') \in E$}
				{
					$\delta(u)(u'') \leftarrow \delta(u)(u'') \cdot f$\;
				}
				
				\ForEach{$u' \in V: (u', u) \in E$}
				{
					\ForEach{$u'' \in V: (u, u'') \in E$}
					{
						$\delta(u')(u'') \leftarrow \delta(u')(u'') + \delta(u')(u) \cdot \delta(u, u'')$\;
						$E \leftarrow E \cup \{ (u', u'') \}$
					}
				}
				\FMain($(V \setminus \{u\}, E, \delta), \target$)\;
				$\hitprob(u, \target) \leftarrow 0$\;
				\ForEach{$u'' \in V : (u, u'') \in E$}
				{
					$\hitprob(u, \target) \leftarrow \hitprob(u, \target) + \delta(u, u'') \cdot \hitprob(u'', \target)$\;
				}
			}
		}
	}
	\vspace{5mm}
	\caption{A Simple Algorithm for Computing Hitting Probabilities.}
	\label{algo:vertex-rem}
\end{algorithm}

\begin{algorithm}
	\vspace{3mm}
	\newcommand{\sol}{\textsl{solution}}
	\DontPrintSemicolon
	\SetKw{True}{true}
	\SetKw{nosol}{Unsatisfiable}
	\SetKw{infsol}{Underdetermined}
	\SetKw{Break}{break}
	\SetKwFunction{FMain}{SolveLinearSystem}
	\SetKwFunction{GS}{Gramm-Schmidt}
	\SetKwBlock{Repeat}{repeat}{}
	\SetKwProg{Fn}{Function}{:}{}
	\Fn{\FMain{$S, G=(V, E), (T, E_T)$}}{
		\eIf{$V = \{\emptyset\}$}
		{
			$\sol \leftarrow \emptyset$\;
			\Return $\sol$\;
		}
		{
			
			\Repeat
			{
				Choose an arbitrary leaf bag $l \in T$\;
				$\bar{l} \leftarrow$ parent of $l$\;
				\eIf{$V_l \subseteq V_{\bar{l}}$}
				{
					$T \leftarrow T \setminus \{l\}$\;
					$E_T \leftarrow E_T \setminus \{ (\bar{l}, l) \}$
				}
				{
					Choose an arbitrary $x \in V_l \setminus V_{\bar{l}}$\;
					$V_l \leftarrow V_l \setminus \{x\}$\;
					\Break
				}
			}
			
			\ForEach{$y_1, y_2 \in V_l: y_1 \neq y_2$}
			{
				$E \leftarrow E \cup \{(y_1, y_2)\}$
			}
			
			$\mathfrak{E} \leftarrow $ equations in $S$ that contain $x$ with non-zero coefficient\;
			$S \leftarrow S \setminus \mathfrak{E}$\;
			
			\If{$\GS(\mathfrak{E}) = \nosol$}{\Return $\nosol$\;}
			
			$\mathfrak{E} \leftarrow \GS(\mathfrak{E})$
			
			\eIf{$\mathfrak{E} = \emptyset$}
			{
				\eIf{$\FMain(S, G \setminus \{x\}, (T, E_T)) = \nosol$}
				{
					\Return $\nosol$\;
				}
				{
					\Return $\infsol$\;
				}
			}
			{
				Choose an arbitrary $\mathfrak{e} \in \mathfrak{E}$ and write it as $x = R_x$\;
				$\mathfrak{E} \leftarrow \mathfrak{E} \setminus \{ \mathfrak{e} \}$\;
				\ForEach{$\mathfrak{e'} \in \mathfrak{E}$}{
					$\mathfrak{e'} \leftarrow \mathfrak{e'}[R_x/x]$ //replace every occurrence of $x$ with $R_x$\;
				}
				$S \leftarrow S \cup \mathfrak{E}$\;
				\eIf{$\FMain(S, G \setminus \{x\}, (T, E_T)) \in \{\nosol, \infsol\}$}
				{
					\Return $\FMain(S, G \setminus \{x\}, (T, E_T))$\;
				}
				{
					$\sol \leftarrow \FMain(S, G \setminus \{x\}, (T, E_T))$\;
					$\sol \leftarrow \sol[x \mapsto [R_x]_\sol]$\;
					\Return $\sol$\;
				}
			}
			
		}
	}
	\vspace{5mm}
	\caption{Solving a system $S$ of linear equations, given its primal graph $G = (V, E)$ and exploiting a tree decomposition $(T, E_T)$ of $G$. Note that $G$ is undirected. Lines~16--17 ensure that $G$ always remains a supergraph of the primal graph of $S$ and that $(T, E_T)$ always remains a valid tree decomposition of $G$.}
	\label{algo:system}
\end{algorithm}

\newpage

\section{Details of Experimental Results} \label{app:exp}

\paragraph{Experimental Setting.} The results were obtained on Ubuntu 18.04 with an Intel Core i5-7200U processor (2.5 GHz, 4 MB cache) using 8 GB of RAM.

\paragraph{Details about Benchmarks.} Table~\ref{table:benchmarks} provides an overview of the \textsf{DaCapo} benchmarks used in our experimental results.

\begin{table}[H]
	\begin{center}
		\resizebox{\textwidth}{!}{
			\input{benchmark-details}
		}
	\end{center}
	\caption{Details of our benchmarks. In each case, $\vert f \vert$ is the number of functions in the benchmark, $\vert V \vert$ is the total number of vertices and $\vert E \vert$ is the total number of edges. Moreover, $t$ is the width of the tree decomposition constructed by JTDec~\cite{jtdec}. Note that this is an upper-bound on the treewidth, given that JTDec is not an exact tool.}
	\label{table:benchmarks}
\end{table}

\paragraph{Raw Numbers and Details of Experimental Results.} Tables~\ref{table:exp-mc-hitting}--\ref{table:exp-mdp-ds} provide runtimes of each of the approaches mentioned in Section~\ref{sec:exp} over every benchmark.

\begin{table}[H]
	\begin{center}
	\input{table-mc-results-hitting}

	\end{center}
	\vspace{1em}
	\caption{Detailed Experimental Results for \textbf{Hitting Probabilities in MCs}. All runtimes are reported in seconds. Note that \textsf{Prism} and \textsf{Storm} round the times to the nearest millisecond, while \textsf{Matlab} rounds to the nearest centisecond.}
	\label{table:exp-mc-hitting}
\end{table} 

\begin{table}[H]
	\begin{center}
		\input{table-mc-results-ds}

	\end{center}
	\vspace{1em}
	\caption{Detailed Experimental Results for \textbf{Expected Discounted Sums in MCs}. All runtimes are reported in seconds. Note that \textsf{Prism} and \textsf{Storm} round the times to the nearest millisecond, while \textsf{Matlab} rounds to the nearest centisecond.}
	\label{table:exp-mc-ds}
\end{table}

\begin{table}[H]
	\hspace{-1.5cm}
	\input{table-mdp-results-hitting}
	\vspace{1em}
	\caption{Detailed Experimental Results for \textbf{Hitting Probabilities in MDPs}. All runtimes are reported in seconds. Note that \textsf{Prism} and \textsf{Storm} round the times to the nearest millisecond, while \textsf{Matlab} rounds to the nearest centisecond.}
	\label{table:exp-mdp-hitting}
\end{table} 

\begin{table}[H]
	\input{table-mdp-results-ds}
	\vspace{1em}
	\caption{Detailed Experimental Results for \textbf{Expected Discounted Sums in MDPs}. All runtimes are reported in seconds. Note that \textsf{Prism} and \textsf{Storm} round the times to the nearest millisecond, while \textsf{Matlab} rounds to the nearest centisecond.}
	\label{table:exp-mdp-ds}
\end{table} 

\paragraph{Remark.} As mentioned before, our inputs contain tree decompositions of the MCs/MDPs. Note that the time used to compute the tree decompositions is negligible, given that constant-width tree decompositions of CFGs are computed by a single pass of the program parse tree~\cite{thorup1998all,jtdec}.

%% file: benchmark-details.tex
\begin{tabular}{|r||c|c|c|c|}
	\hline
	\textbf{Benchmark} & $\vert f \vert$ & $\vert V \vert$ & $\vert E \vert$ & $~~t~~$\\ \hline \hline
asm-3.1 & 105 & 3044 & 3262 & 4 \\ \hline
asm-commons-3.1 & 168 & 2404 & 2473 & 9 \\ \hline
avalon-framework-4.2.0 & 153 & 1899 & 1849 & 4 \\ \hline
avrora-cvs-20091224 & 2539 & 43685 & 43521 & 9 \\ \hline
bootstrap & 29 & 936 & 967 & 5 \\ \hline
commons-codec & 146 & 2728 & 2973 & 5 \\ \hline
commons-daemon & 28 & 453 & 437 & 4 \\ \hline
commons-httpclient & 693 & 9765 & 9772 & 5 \\ \hline
commons-io-1.3.1 & 216 & 3216 & 3175 & 5 \\ \hline
commons-logging & 106 & 2231 & 2303 & 4 \\ \hline
commons-logging-1.0.4 & 53 & 689 & 677 & 3 \\ \hline
constantine & 34 & 776 & 758 & 4 \\ \hline
crimson-1.1.3 & 378 & 8572 & 9328 & 8 \\ \hline
dacapo-digest & 8 & 201 & 208 & 3 \\ \hline
dacapo-h2 & 57 & 1293 & 1311 & 9 \\ \hline
dacapo-luindex & 3 & 84 & 87 & 4 \\ \hline
dacapo-lusearch & 5 & 282 & 300 & 4 \\ \hline
dacapo-lusearch-fix & 5 & 282 & 300 & 4 \\ \hline
dacapo-tomcat & 18 & 250 & 244 & 3 \\ \hline
dacapo-xalan & 10 & 219 & 216 & 3
 \\ \hline
\end{tabular}

\begin{tabular}{|r||c|c|c|c|}
	\hline
	\textbf{Benchmark} & $\vert f \vert$ & $\vert V \vert$ & $\vert E \vert$ & $~~t~~$\\ \hline \hline
	daytrader & 12 & 339 & 332 & 3 \\ \hline
	derbyclient & 2097 & 37865 & 37997 & 9 \\ \hline
	eclipse & 1974 & 45657 & 47039 & 8 \\ \hline
	jaffl & 455 & 6099 & 6126 & 9 \\ \hline
	janino-2.5.15 & 942 & 16861 & 17021 & 8 \\ \hline
	jaxen-1.1.1 & 425 & 5490 & 5375 & 5 \\ \hline
	jline-0.9.95 & 209 & 2427 & 2387 & 5 \\ \hline
	jnr-posix & 165 & 2040 & 1902 & 4 \\ \hline
	junit-3.8.1 & 453 & 4356 & 4067 & 5 \\ \hline
	lucene-core-2.4 & 1216 & 24906 & 25795 & 6 \\ \hline
	lucene-demos-2.4 & 120 & 4063 & 4413 & 7 \\ \hline
	pmd-4.2.5 & 2131 & 37822 & 38672 & 7 \\ \hline
	serializer & 465 & 11038 & 11751 & 6 \\ \hline
	serializer-2.7.0 & 330 & 6174 & 6447 & 9 \\ \hline
	tomcat-juli & 45 & 738 & 740 & 5 \\ \hline
	xalan-2.6.0 & 2088 & 35765 & 36946 & 8 \\ \hline
	xerces\_2\_5\_0 & 2129 & 50279 & 53639 & 9 \\ \hline
	xml-apis & 5 & 19 & 14 & 1 \\ \hline
	xml-apis-1.3.04 & 5 & 19 & 14 & 1 \\ \hline
	xmlgraphics-1.3.1 & 1014 & 17677 & 17890 & 9 \\ \hline
\end{tabular}

%% file: table-mc-results-hitting.tex
\begin{tabular}{r||c|c|c|c|c|cHHHH}
	\hline
	\multirow{2}{*}{Benchmark} & \multicolumn{6}{c}{Runtime in seconds} &  \\ \cline{2-11} 
	& Ours        & \textsf{Gauss}        & \textsf{Matlab} & \textsf{Gurobi} & \textsf{Prism} & \textsf{Storm}  & Ours & \textsf{Gauss} & \textsf{Matlab} & \textsf{Gurobi}\\ \hline \hline
	xml-apis   &  0.00001  &  0.00001  &  0.20000  &  0.04082  &  0.06000  &  0.01000  &  0.00002  &  0.00001  &  0.09000  &  0.04552\\ \hline 
	xml-apis-1.3.04   &  0.00001  &  0.00001  &  0.13000  &  0.03398  &  0.06200  &  0.02000  &  0.00002  &  0.00001  &  0.11000  &  0.03570\\ \hline 
	dacapo-luindex   &  0.00008  &  0.00020  &  0.06000  &  0.01897  &  0.11800  &  0.01200  &  0.00010  &  0.00020  &  0.09000  &  0.03646\\ \hline 
	dacapo-digest   &  0.00017  &  0.00063  &  0.11000  &  0.08734  &  0.30900  &  0.03000  &  0.00023  &  0.00062  &  0.12000  &  0.05931\\ \hline 
	dacapo-xalan   &  0.00017  &  0.00084  &  0.15000  &  0.11467  &  0.37300  &  0.04000  &  0.00022  &  0.00090  &  0.09000  &  0.05408\\ \hline 
	dacapo-tomcat   &  0.00017  &  0.00062  &  0.17000  &  0.11905  &  0.51400  &  0.05000  &  0.00024  &  0.00059  &  0.13000  &  0.16045\\ \hline 
	dacapo-lusearch   &  0.00030  &  0.00566  &  0.11000  &  0.04260  &  0.57600  &  0.03100  &  0.00039  &  0.00584  &  0.10000  &  0.02560\\ \hline 
	dacapo-lusearch-fix   &  0.00034  &  0.00572  &  0.10000  &  0.03674  &  0.42700  &  0.03200  &  0.00040  &  0.00572  &  0.12000  &  0.04174\\ \hline 
	daytrader   &  0.00020  &  0.00130  &  0.11000  &  0.10136  &  0.66800  &  0.04700  &  0.00028  &  0.00141  &  0.15000  &  0.06834\\ \hline 
	commons-daemon   &  0.00033  &  0.00139  &  0.17000  &  0.21845  &  1.48100  &  0.08400  &  0.00043  &  0.00139  &  0.18000  &  0.18583\\ \hline 
	commons-logging-1.0.4   &  0.00048  &  0.00218  &  0.37000  &  0.42445  &  1.49000  &  0.15000  &  0.00067  &  0.00214  &  0.83000  &  0.35145\\ \hline 
	tomcat-juli   &  0.00058  &  0.00371  &  0.28000  &  0.31819  &  1.61300  &  0.19600  &  0.00078  &  0.00370  &  0.36000  &  0.39547\\ \hline 
	constantine   &  0.00060  &  0.03897  &  0.29000  &  0.25880  &  1.68700  &  0.12900  &  0.00082  &  0.03912  &  0.32000  &  0.27186\\ \hline 
	bootstrap   &  0.00084  &  0.03084  &  0.30000  &  0.22259  &  2.49600  &  0.13000  &  0.00111  &  0.03096  &  0.31000  &  0.22324\\ \hline 
	dacapo-h2   &  0.00104  &  0.01446  &  2.55000  &  0.46082  &  2.22300  &  0.20500  &  0.00160  &  0.01178  &  5.81000  &  0.47357\\ \hline 
	avalon-framework-4.2.0   &  0.00139  &  0.00359  &  1.08000  &  1.16564  &  4.48500  &  0.46400  &  0.00190  &  0.00278  &  1.12000  &  1.24184\\ \hline 
	jnr-posix   &  0.00142  &  0.25385  &  1.23000  &  1.13165  &  7.02200  &  0.78700  &  0.00194  &  0.25623  &  1.42000  &  1.28773\\ \hline 
	commons-logging   &  0.00189  &  0.04127  &  0.84000  &  0.81669  &  4.25500  &  0.38800  &  0.00255  &  0.03683  &  0.42000  &  0.88823\\ \hline 
	jline-0.9.95-SNAPSHOT   &  0.00188  &  0.00793  &  1.41000  &  1.39634  &  8.37200  &  0.91000  &  0.00258  &  0.00788  &  1.65000  &  1.49695\\ \hline 
	asm-commons-3.1   &  0.00217  &  0.03162  &  1.16000  &  1.26295  &  6.18000  &  0.50400  &  0.00288  &  0.03166  &  1.42000  &  1.29334\\ \hline 
	commons-codec   &  0.00257  &  0.03337  &  1.09000  &  0.99834  &  7.50300  &  0.51300  &  0.00398  &  0.03401  &  1.24000  &  0.97073\\ \hline 
	commons-io-1.3.1   &  0.00250  &  0.01205  &  1.36000  &  1.59909  &  8.51900  &  0.68300  &  0.00358  &  0.01210  &  1.22000  &  1.49369\\ \hline 
	asm-3.1   &  0.00285  &  0.13666  &  0.89000  &  0.78284  &  7.46700  &  0.42700  &  0.00380  &  0.13651  &  0.92000  &  0.80580\\ \hline 
	junit-3.8.1   &  0.00322  &  0.00493  &  1.83000  &  3.10246  &  18.74200  &  1.80400  &  0.00423  &  0.00485  &  1.85000  &  3.80143\\ \hline 
	lucene-demos-2.4   &  0.00388  &  0.10829  &  1.09000  &  0.88066  &  7.77000  &  0.67400  &  0.00526  &  0.10898  &  1.07000  &  1.13659\\ \hline 
	jaxen-1.1.1   &  0.00419  &  0.60061  &  2.73000  &  3.44220  &  17.16700  &  1.88600  &  0.00563  &  0.60335  &  2.30000  &  3.22225\\ \hline 
	jaffl   &  0.00522  &  0.09861  &  2.58000  &  3.19677  &  16.49900  &  2.30600  &  0.00706  &  0.09858  &  2.41000  &  3.27322\\ \hline 
	serializer-2.7.0   &  0.00543  &  0.14564  &  1.82000  &  2.36092  &  13.21600  &  1.46200  &  0.00715  &  0.14182  &  2.03000  &  2.86329\\ \hline 
	crimson-1.1.3   &  0.00838  &  0.13650  &  2.48000  &  2.84511  &  17.47500  &  1.43200  &  0.01122  &  0.12941  &  2.25000  &  2.83951\\ \hline 
	commons-httpclient   &  0.00782  &  0.05948  &  4.01000  &  5.60121  &  25.39700  &  2.15900  &  0.01070  &  0.05777  &  3.71000  &  4.89463\\ \hline 
	serializer   &  0.01048  &  0.30260  &  2.63000  &  3.32304  &  32.41300  &  2.59700  &  0.01367  &  0.30280  &  2.60000  &  3.53548\\ \hline 
	janino-2.5.15   &  0.01485  &  0.27481  &  5.38000  &  7.09017  &  57.19600  &  7.69100  &  0.01970  &  0.28066  &  4.35000  &  7.31999\\ \hline 
	xmlgraphics-commons-1.3.1   &  0.01621  &  0.18939  &  5.87000  &  7.08701  &  40.47300  &  5.93600  &  0.02116  &  0.18956  &  6.91000  &  8.74569\\ \hline 
	lucene-core-2.4   &  0.02834  &  1.02118  &  5.16000  &  9.00895  &  71.27600  &  6.56200  &  0.03520  &  0.97595  &  5.90000  &  9.92471\\ \hline 
	xalan-2.6.0   &  0.03358  &  1.77685  &  0.87000  &  15.09464  &  95.42400  &  12.41900  &  0.04760  &  1.77619  &  0.70000  &  17.70531\\ \hline 
	derbyclient   &  0.03198  &  0.59163  &  8.02000  &  16.32040  &  81.51000  &  7.72300  &  0.04193  &  0.59825  &  10.19000  &  16.99502\\ \hline 
	pmd-4.2.5   &  0.03257  &  0.55702  &  9.54000  &  15.22455  &  82.97800  &  10.46600  &  0.04358  &  0.55667  &  10.01000  &  16.07630\\ \hline 
	avrora-cvs-20091224   &  0.03832  &  0.68720  &  3.02000  &  19.21516  &  139.14800  &  8.29700  &  0.05008  &  0.66573  &  2.40000  &  17.86095\\ \hline 
	eclipse   &  0.03693  &  0.50101  &  10.76000  &  16.24855  &  83.30500  &  12.38800  &  0.05109  &  0.49472  &  12.60000  &  14.92674\\ \hline 
	xerces\_2\_5\_0   &  0.05558  &  2.53732  &  12.11000  &  16.05042  &  159.37600  &  22.75000  &  0.07604  &  2.42570  &  13.31000  &  18.45862\\ \hline
\end{tabular}

%% file: table-mc-results-ds.tex
\begin{tabular}{r||HHHHHHc|c|c|c}
	\hline
	\multirow{2}{*}{Benchmark} & \multicolumn{6}{H}{Hitting Probabilities} & \multicolumn{4}{c}{Runtime in seconds} \\ \cline{2-11} 
	& Ours        & \textsf{Gauss}        & \textsf{Matlab} & \textsf{Gurobi} & \textsf{Prism} & \textsf{Storm}  & Ours & \textsf{Gauss} & \textsf{Matlab} & \textsf{Gurobi}\\ \hline \hline
	xml-apis   &  0.00001  &  0.00001  &  0.20000  &  0.04082  &  0.06000  &  0.01000  &  0.00002  &  0.00001  &  0.09000  &  0.04552\\ \hline 
	xml-apis-1.3.04   &  0.00001  &  0.00001  &  0.13000  &  0.03398  &  0.06200  &  0.02000  &  0.00002  &  0.00001  &  0.11000  &  0.03570\\ \hline 
	dacapo-luindex   &  0.00008  &  0.00020  &  0.06000  &  0.01897  &  0.11800  &  0.01200  &  0.00010  &  0.00020  &  0.09000  &  0.03646\\ \hline 
	dacapo-digest   &  0.00017  &  0.00063  &  0.11000  &  0.08734  &  0.30900  &  0.03000  &  0.00023  &  0.00062  &  0.12000  &  0.05931\\ \hline 
	dacapo-xalan   &  0.00017  &  0.00084  &  0.15000  &  0.11467  &  0.37300  &  0.04000  &  0.00022  &  0.00090  &  0.09000  &  0.05408\\ \hline 
	dacapo-tomcat   &  0.00017  &  0.00062  &  0.17000  &  0.11905  &  0.51400  &  0.05000  &  0.00024  &  0.00059  &  0.13000  &  0.16045\\ \hline 
	dacapo-lusearch   &  0.00030  &  0.00566  &  0.11000  &  0.04260  &  0.57600  &  0.03100  &  0.00039  &  0.00584  &  0.10000  &  0.02560\\ \hline 
	dacapo-lusearch-fix   &  0.00034  &  0.00572  &  0.10000  &  0.03674  &  0.42700  &  0.03200  &  0.00040  &  0.00572  &  0.12000  &  0.04174\\ \hline 
	daytrader   &  0.00020  &  0.00130  &  0.11000  &  0.10136  &  0.66800  &  0.04700  &  0.00028  &  0.00141  &  0.15000  &  0.06834\\ \hline 
	commons-daemon   &  0.00033  &  0.00139  &  0.17000  &  0.21845  &  1.48100  &  0.08400  &  0.00043  &  0.00139  &  0.18000  &  0.18583\\ \hline 
	commons-logging-1.0.4   &  0.00048  &  0.00218  &  0.37000  &  0.42445  &  1.49000  &  0.15000  &  0.00067  &  0.00214  &  0.83000  &  0.35145\\ \hline 
	tomcat-juli   &  0.00058  &  0.00371  &  0.28000  &  0.31819  &  1.61300  &  0.19600  &  0.00078  &  0.00370  &  0.36000  &  0.39547\\ \hline 
	constantine   &  0.00060  &  0.03897  &  0.29000  &  0.25880  &  1.68700  &  0.12900  &  0.00082  &  0.03912  &  0.32000  &  0.27186\\ \hline 
	bootstrap   &  0.00084  &  0.03084  &  0.30000  &  0.22259  &  2.49600  &  0.13000  &  0.00111  &  0.03096  &  0.31000  &  0.22324\\ \hline 
	dacapo-h2   &  0.00104  &  0.01446  &  2.55000  &  0.46082  &  2.22300  &  0.20500  &  0.00160  &  0.01178  &  5.81000  &  0.47357\\ \hline 
	avalon-framework-4.2.0   &  0.00139  &  0.00359  &  1.08000  &  1.16564  &  4.48500  &  0.46400  &  0.00190  &  0.00278  &  1.12000  &  1.24184\\ \hline 
	jnr-posix   &  0.00142  &  0.25385  &  1.23000  &  1.13165  &  7.02200  &  0.78700  &  0.00194  &  0.25623  &  1.42000  &  1.28773\\ \hline 
	commons-logging   &  0.00189  &  0.04127  &  0.84000  &  0.81669  &  4.25500  &  0.38800  &  0.00255  &  0.03683  &  0.42000  &  0.88823\\ \hline 
	jline-0.9.95-SNAPSHOT   &  0.00188  &  0.00793  &  1.41000  &  1.39634  &  8.37200  &  0.91000  &  0.00258  &  0.00788  &  1.65000  &  1.49695\\ \hline 
	asm-commons-3.1   &  0.00217  &  0.03162  &  1.16000  &  1.26295  &  6.18000  &  0.50400  &  0.00288  &  0.03166  &  1.42000  &  1.29334\\ \hline 
	commons-codec   &  0.00257  &  0.03337  &  1.09000  &  0.99834  &  7.50300  &  0.51300  &  0.00398  &  0.03401  &  1.24000  &  0.97073\\ \hline 
	commons-io-1.3.1   &  0.00250  &  0.01205  &  1.36000  &  1.59909  &  8.51900  &  0.68300  &  0.00358  &  0.01210  &  1.22000  &  1.49369\\ \hline 
	asm-3.1   &  0.00285  &  0.13666  &  0.89000  &  0.78284  &  7.46700  &  0.42700  &  0.00380  &  0.13651  &  0.92000  &  0.80580\\ \hline 
	junit-3.8.1   &  0.00322  &  0.00493  &  1.83000  &  3.10246  &  18.74200  &  1.80400  &  0.00423  &  0.00485  &  1.85000  &  3.80143\\ \hline 
	lucene-demos-2.4   &  0.00388  &  0.10829  &  1.09000  &  0.88066  &  7.77000  &  0.67400  &  0.00526  &  0.10898  &  1.07000  &  1.13659\\ \hline 
	jaxen-1.1.1   &  0.00419  &  0.60061  &  2.73000  &  3.44220  &  17.16700  &  1.88600  &  0.00563  &  0.60335  &  2.30000  &  3.22225\\ \hline 
	jaffl   &  0.00522  &  0.09861  &  2.58000  &  3.19677  &  16.49900  &  2.30600  &  0.00706  &  0.09858  &  2.41000  &  3.27322\\ \hline 
	serializer-2.7.0   &  0.00543  &  0.14564  &  1.82000  &  2.36092  &  13.21600  &  1.46200  &  0.00715  &  0.14182  &  2.03000  &  2.86329\\ \hline 
	crimson-1.1.3   &  0.00838  &  0.13650  &  2.48000  &  2.84511  &  17.47500  &  1.43200  &  0.01122  &  0.12941  &  2.25000  &  2.83951\\ \hline 
	commons-httpclient   &  0.00782  &  0.05948  &  4.01000  &  5.60121  &  25.39700  &  2.15900  &  0.01070  &  0.05777  &  3.71000  &  4.89463\\ \hline 
	serializer   &  0.01048  &  0.30260  &  2.63000  &  3.32304  &  32.41300  &  2.59700  &  0.01367  &  0.30280  &  2.60000  &  3.53548\\ \hline 
	janino-2.5.15   &  0.01485  &  0.27481  &  5.38000  &  7.09017  &  57.19600  &  7.69100  &  0.01970  &  0.28066  &  4.35000  &  7.31999\\ \hline 
	xmlgraphics-commons-1.3.1   &  0.01621  &  0.18939  &  5.87000  &  7.08701  &  40.47300  &  5.93600  &  0.02116  &  0.18956  &  6.91000  &  8.74569\\ \hline 
	lucene-core-2.4   &  0.02834  &  1.02118  &  5.16000  &  9.00895  &  71.27600  &  6.56200  &  0.03520  &  0.97595  &  5.90000  &  9.92471\\ \hline 
	xalan-2.6.0   &  0.03358  &  1.77685  &  0.87000  &  15.09464  &  95.42400  &  12.41900  &  0.04760  &  1.77619  &  0.70000  &  17.70531\\ \hline 
	derbyclient   &  0.03198  &  0.59163  &  8.02000  &  16.32040  &  81.51000  &  7.72300  &  0.04193  &  0.59825  &  10.19000  &  16.99502\\ \hline 
	pmd-4.2.5   &  0.03257  &  0.55702  &  9.54000  &  15.22455  &  82.97800  &  10.46600  &  0.04358  &  0.55667  &  10.01000  &  16.07630\\ \hline 
	avrora-cvs-20091224   &  0.03832  &  0.68720  &  3.02000  &  19.21516  &  139.14800  &  8.29700  &  0.05008  &  0.66573  &  2.40000  &  17.86095\\ \hline 
	eclipse   &  0.03693  &  0.50101  &  10.76000  &  16.24855  &  83.30500  &  12.38800  &  0.05109  &  0.49472  &  12.60000  &  14.92674\\ \hline 
	xerces\_2\_5\_0   &  0.05558  &  2.53732  &  12.11000  &  16.05042  &  159.37600  &  22.75000  &  0.07604  &  2.42570  &  13.31000  &  18.45862\\ \hline
\end{tabular}

%% file: table-mdp-results-hitting.tex
\begin{center}
\begin{tabular}{r||c|c|c|c|c|c|c|cHHHHHH}
	\hline
	\multirow{2}{*}{Benchmark} & \multicolumn{8}{c}{Runtime in seconds} &  \\ \cline{2-15} 
	& Ours & \textsf{lpsolve} & \textsf{Gurobi} & \textsf{Matlab} & \textsf{SI} & \textsf{VI} & \textsf{Prism} & \textsf{Storm} & Ours & \textsf{lpsolve} & \textsf{Gurobi} & \textsf{Matlab} & \textsf{SI} & \textsf{VI}  \\ \hline \hline
	xml-apis   &  0.00004  &  0.04145  &  0.05033  &  3.06000  &  0.00151  &  0.00072  &  0.07700  &  0.01000  &  0.00015  &  0.03889  &  0.07418  &  2.76000  &  0.00318  &  0.00073\\ \hline 
	xml-apis-1.3.04   &  0.00005  &  0.03836  &  0.04782  &  4.33000  &  0.00157  &  0.00085  &  0.06900  &  0.01200  &  0.00012  &  0.04129  &  0.03201  &  3.11000  &  0.00257  &  0.00076\\ \hline 
	dacapo-luindex   &  0.00033  &  0.02243  &  0.01947  &  5.28000  &  0.00787  &  0.00407  &  0.15300  &  0.00900  &  0.00034  &  0.02059  &  0.02004  &  4.57000  &  0.03685  &  0.00431\\ \hline 
	dacapo-digest   &  0.00101  &  0.05517  &  0.09404  &  8.66000  &  0.01816  &  0.00800  &  0.28000  &  0.02600  &  0.00149  &  0.05713  &  0.07190  &  6.46000  &  0.03654  &  0.00978\\ \hline 
	dacapo-xalan   &  0.00069  &  0.07645  &  0.11382  &  7.83000  &  0.01703  &  0.01001  &  0.33700  &  0.02600  &  0.00507  &  0.07229  &  0.10148  &  5.94000  &  0.01549  &  0.01208\\ \hline 
	dacapo-tomcat   &  0.00180  &  0.12966  &  0.18451  &  9.20000  &  0.01916  &  0.00863  &  0.50000  &  0.04600  &  0.00151  &  0.13067  &  0.18988  &  7.84000  &  0.16863  &  0.01018\\ \hline 
	dacapo-lusearch   &  0.00160  &  0.03950  &  0.04988  &  10.55000  &  0.03473  &  0.01553  &  0.34700  &  0.02400  &  0.00239  &  0.03813  &  0.03714  &  8.60000  &  0.04062  &  0.01653\\ \hline 
	dacapo-lusearch-fix   &  0.00516  &  0.04036  &  0.05746  &  8.58000  &  0.02918  &  0.01441  &  0.33600  &  0.02700  &  0.00248  &  0.03770  &  0.04122  &  7.94000  &  0.04010  &  0.02004\\ \hline 
	daytrader   &  0.00073  &  0.08304  &  0.09799  &  8.95000  &  0.03308  &  0.01531  &  0.48600  &  0.03900  &  0.00134  &  0.09054  &  0.08769  &  8.17000  &  0.04721  &  0.01626\\ \hline 
	commons-daemon   &  0.00271  &  0.22508  &  0.23568  &  9.18000  &  0.03780  &  0.01816  &  1.26900  &  0.08000  &  0.00335  &  0.19132  &  0.29852  &  11.59000  &  0.08002  &  0.01998\\ \hline 
	commons-logging-1.0.4   &  0.00175  &  0.38711  &  0.54011  &  19.48000  &  0.12354  &  0.04529  &  1.96100  &  0.13300  &  0.00382  &  0.33759  &  0.53407  &  12.68000  &  0.15830  &  0.03233\\ \hline 
	tomcat-juli   &  0.00325  &  0.27254  &  0.44136  &  12.79000  &  0.06386  &  0.03181  &  1.37400  &  0.15300  &  0.00558  &  0.37249  &  0.48284  &  11.54000  &  0.15890  &  0.20590\\ \hline 
	constantine   &  0.00349  &  0.24723  &  0.33957  &  17.65000  &  0.09099  &  0.03940  &  1.30000  &  0.10000  &  0.00678  &  0.23486  &  0.28992  &  12.54000  &  0.35064  &  0.05305\\ \hline 
	bootstrap   &  0.00788  &  0.23418  &  0.28990  &  17.02000  &  0.12021  &  0.05957  &  1.41200  &  0.10100  &  0.00762  &  0.20695  &  0.26406  &  15.08000  &  0.15577  &  0.07828\\ \hline 
	dacapo-h2   &  0.00507  &  0.44370  &  0.52646  &  22.39000  &  0.13434  &  0.06553  &  1.91600  &  0.16200  &  0.00838  &  0.42674  &  0.54066  &  20.75000  &  0.16036  &  0.07301\\ \hline 
	avalon-framework-4.2.0   &  0.00505  &  1.08039  &  1.35913  &  27.85000  &  0.15185  &  0.07622  &  4.86800  &  0.45100  &  0.01810  &  1.20308  &  1.50038  &  22.67000  &  0.31613  &  0.08697\\ \hline 
	jnr-posix   &  0.01413  &  1.20555  &  1.52360  &  23.05000  &  0.22817  &  0.11717  &  5.74900  &  0.45500  &  0.01688  &  1.23360  &  1.81516  &  21.60000  &  0.37101  &  0.12226\\ \hline 
	commons-logging   &  0.00904  &  0.80701  &  1.00470  &  29.37000  &  0.25868  &  0.11347  &  3.92900  &  0.31700  &  0.02875  &  0.72560  &  0.98885  &  25.38000  &  1.38934  &  0.13886\\ \hline 
	jline-0.9.95-SNAPSHOT   &  0.01112  &  1.48252  &  2.02085  &  20.62000  &  0.20213  &  0.10417  &  6.44600  &  0.55000  &  0.02444  &  1.68731  &  2.18862  &  32.36000  &  0.46697  &  0.14053\\ \hline 
	asm-commons-3.1   &  0.01002  &  1.20580  &  1.55601  &  21.84000  &  0.22011  &  0.10860  &  7.09800  &  0.45600  &  0.01669  &  1.32881  &  2.02500  &  37.33000  &  0.64163  &  0.13915\\ \hline 
	commons-codec   &  0.01272  &  1.18706  &  1.37049  &  20.03000  &  0.31467  &  0.15661  &  6.59400  &  0.49500  &  0.02478  &  1.01284  &  1.32599  &  37.47000  &  0.62215  &  0.15808\\ \hline 
	commons-io-1.3.1   &  0.01134  &  1.70944  &  2.15621  &  48.41000  &  0.29763  &  0.13720  &  6.84900  &  0.57600  &  0.02544  &  1.49117  &  2.22190  &  25.90000  &  1.05737  &  0.15389\\ \hline 
	asm-3.1   &  0.01298  &  0.75230  &  0.98912  &  17.84000  &  0.38381  &  0.20731  &  7.38900  &  0.36200  &  0.02566  &  0.80930  &  1.20244  &  39.74000  &  1.53952  &  0.66412\\ \hline 
	junit-3.8.1   &  0.01116  &  3.28112  &  4.48130  &  48.02000  &  0.28184  &  0.17565  &  12.72300  &  1.24400  &  0.02967  &  3.43031  &  4.54870  &  54.49000  &  1.61903  &  0.19736\\ \hline 
	lucene-demos-2.4   &  0.01978  &  0.92039  &  1.17712  &  40.06000  &  0.67695  &  0.26995  &  6.27900  &  0.54500  &  0.04212  &  0.98169  &  1.38443  &  41.28000  &  1.68492  &  0.55803\\ \hline 
	jaxen-1.1.1   &  0.02025  &  3.17878  &  4.08971  &  41.38000  &  0.62182  &  0.27318  &  15.28600  &  1.17500  &  0.05828  &  3.11529  &  4.00512  &  55.82000  &  7.40472  &  0.34462\\ \hline 
	jaffl   &  0.03770  &  3.05149  &  4.26184  &  59.08000  &  0.65796  &  0.31335  &  16.27800  &  1.42100  &  0.04769  &  3.31774  &  4.50510  &  60.56000  &  1.83675  &  4.18679\\ \hline 
	serializer-2.7.0   &  0.02506  &  2.09810  &  3.66582  &  52.06000  &  0.60972  &  0.31064  &  11.33300  &  1.05700  &  0.04071  &  2.67593  &  3.72690  &  38.10000  &  1.52886  &  0.43002\\ \hline 
	crimson-1.1.3   &  0.05404  &  2.58277  &  3.86147  &  92.63000  &  1.76648  &  0.64496  &  15.51500  &  1.21000  &  0.10414  &  2.85201  &  3.47488  &  72.86000  &  3.67908  &  0.50956\\ \hline 
	commons-httpclient   &  0.03319  &  5.54198  &  6.69283  &  117.38000  &  1.05233  &  0.46138  &  22.27200  &  1.81400  &  0.06327  &  4.79292  &  6.92428  &  78.85000  &  2.95078  &  0.55254\\ \hline 
	serializer   &  0.07065  &  2.91949  &  4.32873  &  56.01000  &  1.28958  &  0.64523  &  25.73500  &  2.00500  &  0.10003  &  3.85245  &  5.27190  &  64.05000  &  2.57026  &  0.88358\\ \hline 
	janino-2.5.15   &  0.07314  &  6.68790  &  9.91680  &  145.83000  &  2.01809  &  1.13186  &  34.01700  &  3.00600  &  0.13146  &  6.94878  &  9.30685  &  133.24000  &  5.84305  &  1.23715\\ \hline 
	xmlgraphics-commons-1.3.1   &  0.07193  &  8.19395  &  11.86806  &  79.65000  &  1.90528  &  0.92933  &  39.59300  &  4.93900  &  0.15004  &  8.27969  &  11.30344  &  144.15000  &  15.20058  &  1.08951\\ \hline 
	lucene-core-2.4   &  0.16863  &  9.23440  &  12.98207  &  322.52000  &  2.92001  &  1.61256  &  52.34900  &  4.53700  &  0.29733  &  8.68221  &  12.39745  &  189.91000  &  7.40384  &  2.24896\\ \hline 
	xalan-2.6.0   &  0.15804  &  13.00583  &  23.80962  &  275.80000  &  4.77874  &  2.50689  &  77.43000  &  9.28700  &  0.30192  &  16.62097  &  21.80253  &  208.95000  &  15.14658  &  2.59791\\ \hline 
	derbyclient   &  0.12489  &  16.30398  &  20.86866  &  256.40000  &  3.77297  &  1.76492  &  87.46500  &  6.32500  &  0.25004  &  15.75678  &  19.19841  &  303.18000  &  14.63486  &  2.18177\\ \hline 
	pmd-4.2.5   &  0.14238  &  13.77016  &  21.09903  &  512.47000  &  5.27168  &  2.25309  &  84.40400  &  7.07000  &  0.25993  &  17.26150  &  21.83183  &  240.07000  &  17.99462  &  2.84897\\ \hline 
	avrora-cvs-20091224   &  0.13498  &  18.70873  &  24.66722  &  290.39000  &  4.39129  &  2.03788  &  101.28400  &  7.92000  &  0.25995  &  19.79896  &  25.89820  &  389.31000  &  19.05126  &  2.56352\\ \hline 
	eclipse   &  0.25944  &  14.32922  &  20.26574  &  406.15000  &  6.40507  &  2.76358  &  104.00300  &  6.65400  &  0.55813  &  14.74519  &  17.84516  &  598.65000  &  15.26801  &  3.71952\\ \hline 
	xerces\_2\_5\_0   &  0.34422  &  14.79255  &  22.62586  &  257.56000  &  9.89841  &  4.58320  &  126.73300  &  12.73800  &  0.46794  &  17.21913  &  24.70113  &  320.46000  &  38.72352  &  4.46858 \\ \hline
\end{tabular}
\end{center}

%% file: table-mdp-results-ds.tex
\begin{center}
\begin{tabular}{r||HHHHHHHHc|c|c|c|c|c}
	\hline
	\multirow{2}{*}{Benchmark} & \multicolumn{8}{H}{} & \multicolumn{6}{c}{Runtime in seconds} \\ \cline{2-15} 
	& Ours & \textsf{lpsolve} & \textsf{Gurobi} & \textsf{Matlab} & \textsf{SI} & \textsf{VI} & \textsf{Prism} & \textsf{Storm} & Ours & \textsf{lpsolve} & \textsf{Gurobi} & \textsf{Matlab} & \textsf{SI} & \textsf{VI}  \\ \hline \hline
	xml-apis   &  0.00004  &  0.04145  &  0.05033  &  3.06000  &  0.00151  &  0.00072  &  0.07700  &  0.01000  &  0.00015  &  0.03889  &  0.07418  &  2.76000  &  0.00318  &  0.00073\\ \hline 
	xml-apis-1.3.04   &  0.00005  &  0.03836  &  0.04782  &  4.33000  &  0.00157  &  0.00085  &  0.06900  &  0.01200  &  0.00012  &  0.04129  &  0.03201  &  3.11000  &  0.00257  &  0.00076\\ \hline 
	dacapo-luindex   &  0.00033  &  0.02243  &  0.01947  &  5.28000  &  0.00787  &  0.00407  &  0.15300  &  0.00900  &  0.00034  &  0.02059  &  0.02004  &  4.57000  &  0.03685  &  0.00431\\ \hline 
	dacapo-digest   &  0.00101  &  0.05517  &  0.09404  &  8.66000  &  0.01816  &  0.00800  &  0.28000  &  0.02600  &  0.00149  &  0.05713  &  0.07190  &  6.46000  &  0.03654  &  0.00978\\ \hline 
	dacapo-xalan   &  0.00069  &  0.07645  &  0.11382  &  7.83000  &  0.01703  &  0.01001  &  0.33700  &  0.02600  &  0.00507  &  0.07229  &  0.10148  &  5.94000  &  0.01549  &  0.01208\\ \hline 
	dacapo-tomcat   &  0.00180  &  0.12966  &  0.18451  &  9.20000  &  0.01916  &  0.00863  &  0.50000  &  0.04600  &  0.00151  &  0.13067  &  0.18988  &  7.84000  &  0.16863  &  0.01018\\ \hline 
	dacapo-lusearch   &  0.00160  &  0.03950  &  0.04988  &  10.55000  &  0.03473  &  0.01553  &  0.34700  &  0.02400  &  0.00239  &  0.03813  &  0.03714  &  8.60000  &  0.04062  &  0.01653\\ \hline 
	dacapo-lusearch-fix   &  0.00516  &  0.04036  &  0.05746  &  8.58000  &  0.02918  &  0.01441  &  0.33600  &  0.02700  &  0.00248  &  0.03770  &  0.04122  &  7.94000  &  0.04010  &  0.02004\\ \hline 
	daytrader   &  0.00073  &  0.08304  &  0.09799  &  8.95000  &  0.03308  &  0.01531  &  0.48600  &  0.03900  &  0.00134  &  0.09054  &  0.08769  &  8.17000  &  0.04721  &  0.01626\\ \hline 
	commons-daemon   &  0.00271  &  0.22508  &  0.23568  &  9.18000  &  0.03780  &  0.01816  &  1.26900  &  0.08000  &  0.00335  &  0.19132  &  0.29852  &  11.59000  &  0.08002  &  0.01998\\ \hline 
	commons-logging-1.0.4   &  0.00175  &  0.38711  &  0.54011  &  19.48000  &  0.12354  &  0.04529  &  1.96100  &  0.13300  &  0.00382  &  0.33759  &  0.53407  &  12.68000  &  0.15830  &  0.03233\\ \hline 
	tomcat-juli   &  0.00325  &  0.27254  &  0.44136  &  12.79000  &  0.06386  &  0.03181  &  1.37400  &  0.15300  &  0.00558  &  0.37249  &  0.48284  &  11.54000  &  0.15890  &  0.20590\\ \hline 
	constantine   &  0.00349  &  0.24723  &  0.33957  &  17.65000  &  0.09099  &  0.03940  &  1.30000  &  0.10000  &  0.00678  &  0.23486  &  0.28992  &  12.54000  &  0.35064  &  0.05305\\ \hline 
	bootstrap   &  0.00788  &  0.23418  &  0.28990  &  17.02000  &  0.12021  &  0.05957  &  1.41200  &  0.10100  &  0.00762  &  0.20695  &  0.26406  &  15.08000  &  0.15577  &  0.07828\\ \hline 
	dacapo-h2   &  0.00507  &  0.44370  &  0.52646  &  22.39000  &  0.13434  &  0.06553  &  1.91600  &  0.16200  &  0.00838  &  0.42674  &  0.54066  &  20.75000  &  0.16036  &  0.07301\\ \hline 
	avalon-framework-4.2.0   &  0.00505  &  1.08039  &  1.35913  &  27.85000  &  0.15185  &  0.07622  &  4.86800  &  0.45100  &  0.01810  &  1.20308  &  1.50038  &  22.67000  &  0.31613  &  0.08697\\ \hline 
	jnr-posix   &  0.01413  &  1.20555  &  1.52360  &  23.05000  &  0.22817  &  0.11717  &  5.74900  &  0.45500  &  0.01688  &  1.23360  &  1.81516  &  21.60000  &  0.37101  &  0.12226\\ \hline 
	commons-logging   &  0.00904  &  0.80701  &  1.00470  &  29.37000  &  0.25868  &  0.11347  &  3.92900  &  0.31700  &  0.02875  &  0.72560  &  0.98885  &  25.38000  &  1.38934  &  0.13886\\ \hline 
	jline-0.9.95-SNAPSHOT   &  0.01112  &  1.48252  &  2.02085  &  20.62000  &  0.20213  &  0.10417  &  6.44600  &  0.55000  &  0.02444  &  1.68731  &  2.18862  &  32.36000  &  0.46697  &  0.14053\\ \hline 
	asm-commons-3.1   &  0.01002  &  1.20580  &  1.55601  &  21.84000  &  0.22011  &  0.10860  &  7.09800  &  0.45600  &  0.01669  &  1.32881  &  2.02500  &  37.33000  &  0.64163  &  0.13915\\ \hline 
	commons-codec   &  0.01272  &  1.18706  &  1.37049  &  20.03000  &  0.31467  &  0.15661  &  6.59400  &  0.49500  &  0.02478  &  1.01284  &  1.32599  &  37.47000  &  0.62215  &  0.15808\\ \hline 
	commons-io-1.3.1   &  0.01134  &  1.70944  &  2.15621  &  48.41000  &  0.29763  &  0.13720  &  6.84900  &  0.57600  &  0.02544  &  1.49117  &  2.22190  &  25.90000  &  1.05737  &  0.15389\\ \hline 
	asm-3.1   &  0.01298  &  0.75230  &  0.98912  &  17.84000  &  0.38381  &  0.20731  &  7.38900  &  0.36200  &  0.02566  &  0.80930  &  1.20244  &  39.74000  &  1.53952  &  0.66412\\ \hline 
	junit-3.8.1   &  0.01116  &  3.28112  &  4.48130  &  48.02000  &  0.28184  &  0.17565  &  12.72300  &  1.24400  &  0.02967  &  3.43031  &  4.54870  &  54.49000  &  1.61903  &  0.19736\\ \hline 
	lucene-demos-2.4   &  0.01978  &  0.92039  &  1.17712  &  40.06000  &  0.67695  &  0.26995  &  6.27900  &  0.54500  &  0.04212  &  0.98169  &  1.38443  &  41.28000  &  1.68492  &  0.55803\\ \hline 
	jaxen-1.1.1   &  0.02025  &  3.17878  &  4.08971  &  41.38000  &  0.62182  &  0.27318  &  15.28600  &  1.17500  &  0.05828  &  3.11529  &  4.00512  &  55.82000  &  7.40472  &  0.34462\\ \hline 
	jaffl   &  0.03770  &  3.05149  &  4.26184  &  59.08000  &  0.65796  &  0.31335  &  16.27800  &  1.42100  &  0.04769  &  3.31774  &  4.50510  &  60.56000  &  1.83675  &  4.18679\\ \hline 
	serializer-2.7.0   &  0.02506  &  2.09810  &  3.66582  &  52.06000  &  0.60972  &  0.31064  &  11.33300  &  1.05700  &  0.04071  &  2.67593  &  3.72690  &  38.10000  &  1.52886  &  0.43002\\ \hline 
	crimson-1.1.3   &  0.05404  &  2.58277  &  3.86147  &  92.63000  &  1.76648  &  0.64496  &  15.51500  &  1.21000  &  0.10414  &  2.85201  &  3.47488  &  72.86000  &  3.67908  &  0.50956\\ \hline 
	commons-httpclient   &  0.03319  &  5.54198  &  6.69283  &  117.38000  &  1.05233  &  0.46138  &  22.27200  &  1.81400  &  0.06327  &  4.79292  &  6.92428  &  78.85000  &  2.95078  &  0.55254\\ \hline 
	serializer   &  0.07065  &  2.91949  &  4.32873  &  56.01000  &  1.28958  &  0.64523  &  25.73500  &  2.00500  &  0.10003  &  3.85245  &  5.27190  &  64.05000  &  2.57026  &  0.88358\\ \hline 
	janino-2.5.15   &  0.07314  &  6.68790  &  9.91680  &  145.83000  &  2.01809  &  1.13186  &  34.01700  &  3.00600  &  0.13146  &  6.94878  &  9.30685  &  133.24000  &  5.84305  &  1.23715\\ \hline 
	xmlgraphics-commons-1.3.1   &  0.07193  &  8.19395  &  11.86806  &  79.65000  &  1.90528  &  0.92933  &  39.59300  &  4.93900  &  0.15004  &  8.27969  &  11.30344  &  144.15000  &  15.20058  &  1.08951\\ \hline 
	lucene-core-2.4   &  0.16863  &  9.23440  &  12.98207  &  322.52000  &  2.92001  &  1.61256  &  52.34900  &  4.53700  &  0.29733  &  8.68221  &  12.39745  &  189.91000  &  7.40384  &  2.24896\\ \hline 
	xalan-2.6.0   &  0.15804  &  13.00583  &  23.80962  &  275.80000  &  4.77874  &  2.50689  &  77.43000  &  9.28700  &  0.30192  &  16.62097  &  21.80253  &  208.95000  &  15.14658  &  2.59791\\ \hline 
	derbyclient   &  0.12489  &  16.30398  &  20.86866  &  256.40000  &  3.77297  &  1.76492  &  87.46500  &  6.32500  &  0.25004  &  15.75678  &  19.19841  &  303.18000  &  14.63486  &  2.18177\\ \hline 
	pmd-4.2.5   &  0.14238  &  13.77016  &  21.09903  &  512.47000  &  5.27168  &  2.25309  &  84.40400  &  7.07000  &  0.25993  &  17.26150  &  21.83183  &  240.07000  &  17.99462  &  2.84897\\ \hline 
	avrora-cvs-20091224   &  0.13498  &  18.70873  &  24.66722  &  290.39000  &  4.39129  &  2.03788  &  101.28400  &  7.92000  &  0.25995  &  19.79896  &  25.89820  &  389.31000  &  19.05126  &  2.56352\\ \hline 
	eclipse   &  0.25944  &  14.32922  &  20.26574  &  406.15000  &  6.40507  &  2.76358  &  104.00300  &  6.65400  &  0.55813  &  14.74519  &  17.84516  &  598.65000  &  15.26801  &  3.71952\\ \hline 
	xerces\_2\_5\_0   &  0.34422  &  14.79255  &  22.62586  &  257.56000  &  9.89841  &  4.58320  &  126.73300  &  12.73800  &  0.46794  &  17.21913  &  24.70113  &  320.46000  &  38.72352  &  4.46858 \\ \hline
\end{tabular}
\end{center}